\RequirePackage{amsmath} 
\documentclass[runningheads]{llncs}

\usepackage[T1]{fontenc}
\usepackage[symbol]{footmisc} 
\usepackage{tablefootnote}
\usepackage{multirow, makecell}

\spnewtheorem{claim}{Claim}{\itshape}{\rmfamily}

\AtBeginDocument{%
  \DeclareFontShape{T1}{cmr}{m}{scit}{<->ssub*lmr/m/scsl}{}%
}

\usepackage{amsfonts,amssymb}
\usepackage{graphicx}
\usepackage{color}
\usepackage{longtable}
\usepackage{hyperref}
\usepackage[normalem]{ulem}
\usepackage[color=yellow!30]{todonotes}
\usepackage[caption=false]{subfig}
\usepackage{framed}

\usepackage{tikz}
\usetikzlibrary{calc}
\usetikzlibrary{shapes.geometric}
\usetikzlibrary{shapes,decorations,arrows,calc,arrows.meta,fit,positioning}
\tikzset{
	-Latex,auto,node distance =1 cm and 1 cm,semithick,
	state/.style ={ellipse, draw, minimum width = 0.7 cm},
	point/.style = {circle, draw, inner sep=0.04cm,fill,node contents={}},
	bidirected/.style={Latex-Latex,dashed},
	el/.style = {inner sep=2pt, align=left, sloped},
	square/.style={regular polygon,regular polygon sides=4}
}

\DeclareFontFamily{U} {MnSymbolA}{}
\DeclareFontShape{U}{MnSymbolA}{m}{n}{
  <-6> MnSymbolA5
  <6-7> MnSymbolA6
  <7-8> MnSymbolA7
  <8-9> MnSymbolA8
  <9-10> MnSymbolA9
  <10-12> MnSymbolA10
  <12-> MnSymbolA12}{}
\DeclareFontShape{U}{MnSymbolA}{b}{n}{
  <-6> MnSymbolA-Bold5
  <6-7> MnSymbolA-Bold6
  <7-8> MnSymbolA-Bold7
  <8-9> MnSymbolA-Bold8
  <9-10> MnSymbolA-Bold9
  <10-12> MnSymbolA-Bold10
  <12-> MnSymbolA-Bold12}{}
\DeclareSymbolFont{MnSyA} {U} {MnSymbolA}{m}{n}
\DeclareMathSymbol{\static}{\mathrel}{MnSyA}{0}
\DeclareMathSymbol{\dynamic}{\mathrel}{MnSyA}{96}

\usepackage{array}

\def\qed{\hfill\hbox{${\vcenter{\vbox{
	\hrule height 0.4pt\hbox{\vrule width 0.4pt height 6pt						
	\kern5pt\vrule width 0.4pt}\hrule height 0.4pt}}}$}}

\makeatletter
\let\save@mathaccent\mathaccent
\newcommand*\if@single[3]{%
	\setbox0\hbox{${\mathaccent"0362{#1}}^H$}%
	\setbox2\hbox{${\mathaccent"0362{\kern0pt#1}}^H$}%
	\ifdim\ht0=\ht2 #3\else #2\fi
}
\newcommand*\rel@kern[1]{\kern#1\dimexpr\macc@kerna}
\newcommand*\widebar[1]{\@ifnextchar^{{\wide@bar{#1}{0}}}{\wide@bar{#1}{1}}}
\newcommand*\wide@bar[2]{\if@single{#1}{\wide@bar@{#1}{#2}{1}}{\wide@bar@{#1}{#2}{2}}}
\newcommand*\wide@bar@[3]{%
	\begingroup
	\def\mathaccent##1##2{%
		\let\mathaccent\save@mathaccent
		\if#32 \let\macc@nucleus\first@char \fi
		\setbox\z@\hbox{$\macc@style{\macc@nucleus}_{}$}%
		\setbox\tw@\hbox{$\macc@style{\macc@nucleus}{}_{}$}%
		\dimen@\wd\tw@
		\advance\dimen@-\wd\z@
		\divide\dimen@ 3
		\@tempdima\wd\tw@
		\advance\@tempdima-\scriptspace
		\divide\@tempdima 10
		\advance\dimen@-\@tempdima
		\ifdim\dimen@>\z@ \dimen@0pt\fi
		\rel@kern{0.6}\kern-\dimen@
		\if#31
		\overline{\rel@kern{-0.6}\kern\dimen@\macc@nucleus\rel@kern{0.4}\kern\dimen@}%
		\advance\dimen@0.4\dimexpr\macc@kerna
		\let\final@kern#2%
		\ifdim\dimen@<\z@ \let\final@kern1\fi
		\if\final@kern1 \kern-\dimen@\fi
		\else
		\overline{\rel@kern{-0.6}\kern\dimen@#1}%
		\fi
	}%
	\macc@depth\@ne
	\let\math@bgroup\@empty \let\math@egroup\macc@set@skewchar
	\mathsurround\z@ \frozen@everymath{\mathgroup\macc@group\relax}%
	\macc@set@skewchar\relax
	\let\mathaccentV\macc@nested@a
	\if#31
	\macc@nested@a\relax111{#1}%
	\else
	\def\gobble@till@marker##1\endmarker{}%
	\futurelet\first@char\gobble@till@marker#1\endmarker
	\ifcat\noexpand\first@char A\else
	\def\first@char{}%
	\fi
	\macc@nested@a\relax111{\first@char}%
	\fi
	\endgroup
}
\makeatother

\newcommand{\moonlit}{\breve{M}}
\newcommand{\sunlit}{\breve{S}}
\newcommand{\planet}{P}
\newcommand{\planetv}{p}
\newcommand{\surface}{\tilde{P}}
\newcommand{\surfacev}{\tilde{p}}
\begin{document}

\title{Reconfigurable routing in data center networks}



\author{David C. Kutner\inst{1}\orcidID{0000-0003-2979-4513} \and
Iain A. Stewart\inst{1}\orcidID{0000-0002-0752-1971}}
\authorrunning{D. C. Kutner and I. A. Stewart}
%
\institute{
Department of Computer Science, Durham University,\\Upper Mountjoy Campus, Stockton Road, Durham DH1 3LE, UK\\
\email{\{david.c.kutner, i.a.stewart\}@durham.ac.uk}}

\maketitle              
\begin{abstract}
A hybrid network is a static (electronic) network that is augmented with optical switches. The Reconfigurable Routing Problem (RRP) in hybrid networks is the problem of finding settings for the optical switches augmenting a static network so as to achieve optimal delivery of some given workload. The problem has previously been studied in various scenarios with both tractability and NP-hardness results obtained. However, the data center and interconnection networks to which the problem is most relevant are almost always such that the static network is highly structured (and often node-symmetric) whereas all previous results assume that the static network can be arbitrary (which makes existing computational hardness results less technologically relevant and also easier to obtain). In this paper, and for the first time, we prove various intractability results for RRP where the underlying static network is highly structured, for example consisting of a hypercube, and also extend some existing tractability results. 
\end{abstract}


\emph{The principal improvement to this paper compared to the previous version is the addition of Subsection \ref{sec:lunar}.}

\section{Introduction}

The rapid growth of cloud computing applications has induced demand for new technologies to optimize the performance of data center networks dealing with ever-larger workloads. 
The data center topology design problem (that of finding efficient data center topologies) has been studied extensively and resulted in myriad designs (see, e.g., \cite{CGC16}).  Advances in hardware, such as optical switches reconfigurable in milli- to micro-seconds, have enabled the development of reconfigurable topologies (see, e.g., \cite{HFS21}). These topologies can adjust in response to demand (\emph{demand-aware} reconfigurable topologies) or vary configurations over time according to a fixed protocol (\emph{demand-oblivious} reconfigurable topologies; see, e.g., \cite{AS19}). So-called \emph{hybrid} data center networks are a combination of a static topology consisting of, for example, electrical switches, and a demand-aware reconfigurable topology implemented, for example, with optical circuit switches or free space optics (see, e.g., \cite{CSBDRBTS18,FS19,firefly14,ZSXXTC22}). An intuitive example of a simple reconfigurable topology is illustrated in Fig. \ref{fig:mirror}. 

\begin{figure}[!ht]
    \centering
    \includegraphics[page=5, width =.8\textwidth]{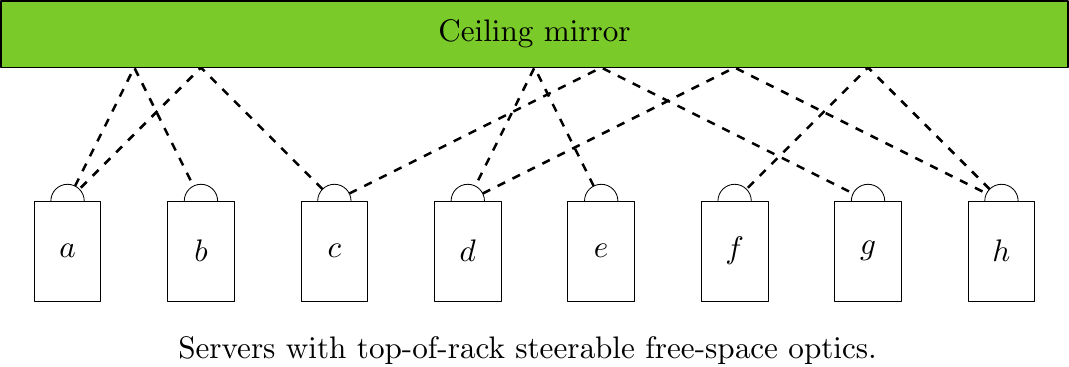}
    \caption{Basic model of an optical wireless data-center network, as described in \cite{CSBDRBTS18,firefly14,ZSXXTC22}. Practical timescales for reconfiguration vary from milliseconds \cite{firefly14} to microseconds or nanoseconds \cite{CSBDRBTS18,ZSXXTC22}.}
    \label{fig:mirror}
\end{figure}

The hybrid network paradigm combines the robustness guarantees of static networks with the ability of demand-aware reconfigurable networks to serve large workloads at very low cost.  
Consider, for example, the hybrid network shown in Fig. \ref{fig:hybrid}, and the configuration shown in Fig. \ref{fig:augmented}. In the (unaugmented) static network, there are two possible paths along which a message from node $b$ to node $d$ may be routed:
$b \static f \static h \static e \static d$ or $b \static f \static e \static d$. In the hybrid network as configured in Fig. \ref{fig:augmented}, the path $b \dynamic a \static c \dynamic d$ (among others) is an option\footnote{We denote by $u\dynamic v$ the concatenation of a switch link from $u$ to some switch, of the internal switch connection, and of a switch link to $v$ from that switch.}. 

\begin{figure}[ht]
	\centering
    \includegraphics[page=2, width=.9\textwidth]{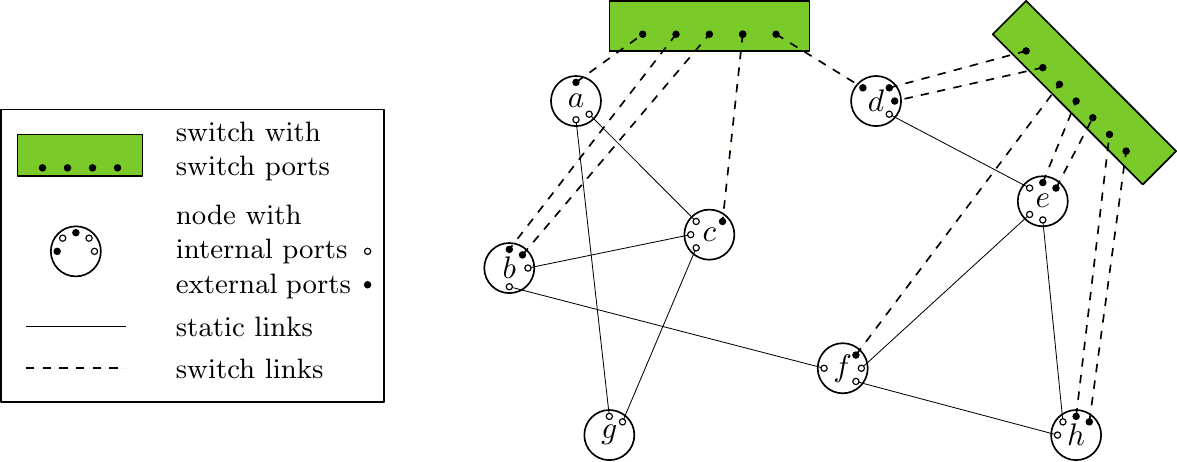}
	\caption{A hybrid network.}\label{fig:hybrid}
\end{figure}

\begin{figure}[ht]
    \includegraphics[page=3, width=.4\textwidth]{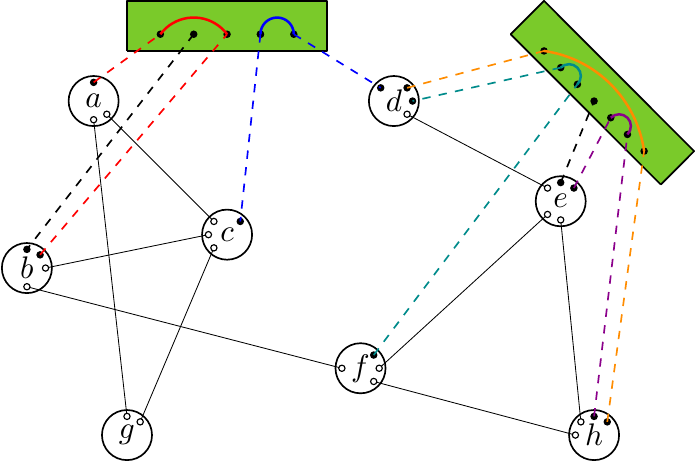}
    \hfill
    \includegraphics[page=4, width=.4\textwidth]{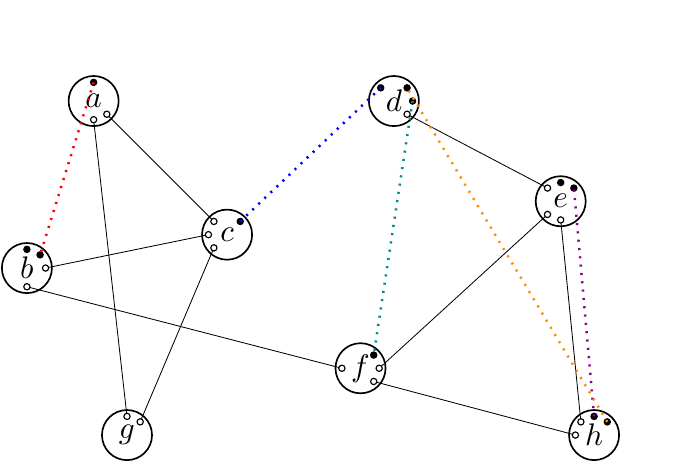}
	
	\caption{An augmented network and its abstracted dynamic links.}\label{fig:augmented}
\end{figure}

Of particular interest to us is the question of how the reconfigurable (optical) portion of the network should be configured for some demand pattern, formalized by Foerster, Ghobadi and Schmid \cite{FGS18} as the \textsc{Reconfigurable Routing Problem} (\textsc{RRP}): in short, given a hybrid network (consisting of a static network and of some switches) and a workload, we wish to choose a \emph{configuration} (setting of the switches) which results in an optimal delivery of the workload.

Crucially, existing hardness results are only valid when the static network is allowed to be arbitrary, which is almost never the case in practice where interconnection and data center network design is driven by symmetry, high connectivity, recursive decomposition, and so forth. For example: the popular switch-centric data center network Fat-Tree \cite{ALV08} is derived from a folded Clos network; the server-centric data center network DCell \cite{GWT08} is recursively-structured whereby at each level, a graph-theoretic matching of servers is imposed; and the server-centric data center network BCube \cite{GLL09} is recursively-structured with a construction based around a generalized hypercube. (It should be noted that there do exist examples of unstructured data center networks, such as Jellyfish \cite{SHP12} and Xpander \cite{VSD16} which utilize the theory of random graphs.) Many (but not all) NP-complete problems become tractable when the input is restricted to the graphs providing the communications fabric for data center networks and other interconnection networks. For example, Hamiltonian paths are often trivial to find in many interconnection networks; indeed, no finite connected vertex-transitive graph \emph{without} a Hamiltonian path is known to exist (the Lovász Conjecture contends there is no such graph - see Section 4 of \cite{PR09}). This motivates our investigation into how the complexity of RRP changes when we restrict to more structured and realistic networks. The question of the complexity of RRP for specific network topologies was specifically identified as an area for future work in \cite{FFSV20}. 

In this paper, we establish for the first time hardness results for RRP that apply to various specific families of highly structured static networks such as, for example, the hypercubes. Our constructions are (perhaps not surprisingly) of a much more involved nature than has hitherto been the case. 

\section{Problem Setting}

The decision problem \textsc{Reconfigurable Routing Problem} considered in this paper is a proper restriction of that presented in prior work \cite{FFSV20,FGS18,FPS19}.
In this section, we provide technical detail to fully formalize our version of the problem, but also additionally provide sufficient framing to briefly review existing results and to identify the areas strengthened by our contribution.






We adopt the usual terminology of graph theory though we tend to use `nodes' and `links' when speaking about the components of reconfigurable networks and `nodes' and `edges' when dealing with (abstract) graphs. We denote the natural numbers by $\mathbb{N}$ (we include $0\in\mathbb{N}$) and the non-negative rationals by $\mathbb{Q}_+$.

\subsection{Hybrid networks, (re)configurations and (segregated) routing}\label{sec:configs-routes-segregation}

A hybrid network $G(S)$ can be visualized as in Fig.~\ref{fig:hybrid}, and consists of a static network $G$ and some switches $S$ augmenting it. 
A \emph{static network\/} $G$ can be abstracted as an undirected graph $G=(V,E)$ so that each \emph{static link\/} $(u,v)\in E$ has some fixed \emph{weight\/} $w\in \mathbb{Q}_+$ (reflecting a transmission cost) and is incident with \emph{internal ports\/} of two distinct nodes of $V$.
The number of internal ports of some node $v\in V$ is then exactly the degree of $v$ in the abstracted graph $G$. 
We denote by $S$ a set of \emph{switches\/} augmenting the static network $G$ with \emph{switch links\/} joining \emph{switch ports\/} of some switch to \emph{external ports\/} of some of the nodes of $V$. 
Every switch link has weight $0$ (we say more about switch link weights momentarily). Every switch $s\in S$ has at least two switch ports.

In general, the number of external ports of the nodes of a static network $G=(V,E)$ is variable, as is the number of switch ports of the switches of a hybrid network $G(S)$, and it may be the case that there is more than one switch link between a specific node and a specific switch. We assume that the switch links describe a bijection between the external ports and the switch ports; otherwise, there would be some unused ports, which we can safely ignore.


Given a hybrid network $G(S)$ and a switch $s\in S$, a \emph{switch matching\/} $N_s$ of $s$ is a set of pairs of switch ports of $s$ so that all switch ports involved are distinct. Each switch matching represents an internal setting of the switch and naturally yields a set of pairs of external ports of nodes where all such ports are distinct; we refer to a set of pairs of external ports obtained in this way as a \emph{node matching\/} (note that this differs from the standard graph-theoretic notion of a matching). An illustration of a configured hybrid network is shown in Fig. \ref{fig:augmented}: on the left side, switch matchings are represented as sets of arcs, and on the right side the corresponding node matching is shown as a set of dotted lines.

A \emph{configuration\/} $N$ is a set of switch matchings, one for each switch. A configuration straightforwardly encodes the corresponding node matchings. We say that $(u,v)$ is a \emph{dynamic link} in the configuration $N$ (we sometimes write $(u,v)\in N$) if $(u,v)$ appears in any node matching corresponding to $N$. 

We allocate a fixed weight $\mu\in \mathbb{Q}_+$ to each internal port-to-port connection in a switch $s$. Although a dynamic link is an atomic entity, it can be visualized as consisting of a switch link followed by an internal port-to-port connection in $s$ followed by another switch link. We denote by $G(N)$ the static network $G$ augmented with the dynamic links (each of weight $\mu$) resulting from the configuration $N$ and we call $G(N)$ an \emph{augmented network\/}. In the augmented network visualized in Fig. \ref{fig:augmented}, for example: $(a,b)$ is a dynamic link; $(a,c)$ is a static link; and $(e,h)$ is both a static link a dynamic link. Note that it is possible that an augmented network $G(N)$ is a multigraph. 

The concepts defined above are driven by reconfigurable hardware technology such as optical switches, wireless (beamforming) and free-space optics, all of which establish port-to-port connections, i.e., switch matchings. The survey paper \cite{FS19} provides some detail as regards the relationship between the emergent theoretical models and current opto-electronic technology.

\subsection{Routing in hybrid networks}

Consider again the example shown in Fig. \ref{fig:augmented}. In the configuration shown, a message $M$ from $c$ to node $e$ may be routed:
\begin{enumerate}
    \item via static links only, along the path $\varphi_1:=c \static b \static f \static e$ with weight $3w$, or\label{opt:static-route}
    \item via dynamic links only, along the path $\varphi_2:=c \dynamic d \dynamic h \dynamic e$ with weight $3\mu$, or\label{opt:dynamic-route}
    \item via a combination of static and dynamic links, along the path $\varphi_3:=c \dynamic d \static e$ with weight $\mu+w$.\label{opt:hybrid-route}
\end{enumerate}

Depending on the value of $\mu$, any of the paths may minimize the cost to route $M$: 
if $\mu \ge 2w$ then $\varphi_1$ is optimal; if $\mu \le \frac{w}{2}$ then $\varphi_2$ is optimal; and if $\mu \in [\frac{w}{2},2w]$ then $\varphi_3$ is an optimal. We may wish to bound the number of alternations allowed between optic and static links in any path a message takes; we capture this hardware requirement via a \emph{segregation parameter} $\sigma \in \mathbb N \cup \{\infty\}$, as introduced in \cite{FPS19}, that is the number of alternations between static and dynamic links. In the fully segregated case, $\sigma=0$: messages may be routed either by static links only (as in $\varphi_1$) or by dynamic links only (as in $\varphi_2$). In the non-segregated case, $\sigma=\infty$ and there is no restriction on the number of alternations, so any path is admitted. Note $\varphi_3$ is admitted as a valid path to route $M$ if and only if $\sigma \ge 1$. The \emph{dynamic link limit} $\delta$, like the segregation parameter $\sigma$, is a restriction on admissible flow-paths. Whereas $\sigma$ describes the maximum number of alternations between static and dynamic links permitted, $\delta$ describes the maximum \emph{number} of dynamic links any flow-path may use. In particular, when $\delta=1$ every flow-path must contain at most one dynamic link.

Networks are expected to route many messages (of varying sizes) optimally at the same time. Given a hybrid network $G(S)$ we represent the set of all demands we must optimize for as a \emph{workload} (\emph{matrix\/}) $D$ with entries $\{D[u,v]\in \mathbb{Q}_+: u,v\in V\}$ providing the intended pairwise \emph{node-to-node workloads} (each $D[u,u]$ is necessarily $0$).

Given a configuration $N$ and $u,v\in V$ for which $D[u,v]>0$, we route the corresponding workload via a path in $G(N)$ from $u$ to $v$ in $G(N)$ so that this chosen \emph{flow-path\/} $\varphi(u,v)$ has \emph{workload cost} $D[u,v]\times wt_{G(N)}(\varphi(u,v))$, where the \emph{weight\/} $wt_{G(N)}(\varphi(u,v))$ is the sum of the weights of the links of the flow-path $\varphi(u,v)$ (if $G(N)$ has both a static link $(x,y)$ and a dynamic link $(x,y)$ then we need to say which we are using in $\varphi(u,v)$). The \emph{total workload cost\/} (of $D$ under $N$) is defined as
    	$$\sum\limits_{u,v\in V, D[u,v]>0}D[u,v]\times wt_{G(N)}(\varphi(u,v)).$$
	
Our aim will be to find a configuration $N$ in some hybrid network $G(S)$ and flow-paths in $G(N)$ for which the total workload cost of some workload matrix $D$ is minimized. In an unrestricted scenario, we would choose any flow-path $\varphi(u,v)$ to be a flow-path of minimum weight from $u$ to $v$ in $G(N)$, the weight of which we denote by $wt_{G(N)}(u,v)$. When $\sigma \ne \infty$ we must also ensure the flow-path has at most $\sigma$ alternations. We also have the analogous concepts $wt_G(\varphi(u,v))$ and $wt_G(u,v)$ where we work entirely in the static network $G$. Note that we often describe $D$ by a weighted digraph, which we usually call $D^\prime$, so that the node set is $V$ and there is an edge $(u,v)$ of weight $w>0$ if, and only if, $D[u,v]=w$. We also refer to some $D[u,v]>0$ as a \emph{demand} (from $u$ to $v$).

\subsection{The Reconfigurable Routing Problem}

We are now in a position to introduce our protagonist:

\begin{framed}
    \noindent
    \textbf{\textsc{Reconfigurable Routing Problem ($\sigma$) (RRP($\sigma$))}}\\
    \emph{Input:} $(G, S, \mu, w, D, \kappa)$: $D$ is a workload matrix for the hybrid network $G(S)$ with static (resp. dynamic) links all of weight $w$ (resp. $\mu$).\\
    \emph{Question:} Does $G(S)$ admit some configuration $N$ such that the total workload cost of $D$ under $N$ (where the number of alternations for any path is bounded by $\sigma$) is at most $\kappa$?
\end{framed}

As previously alluded to, this setting is more expressive than we require for most of this paper, and more restrictive than the exact formalism considered in prior work \cite{FFSV20,FGS18,FPS19}: in those works, $w$ and $\mu$ are sometimes allowed to be functions of their endpoints rather than fixed constants. This provides much more expressivity; notably, their model loses no power when it is restricted to inputs where $G$ is a complete graph and there is only one switch, since it is possible to simulate any other instance by assigning prohibitively large weights to any static edges and any pair of switch ports which should not be usable. 

We now turn to the ``realistic'' networks we mentioned in our introduction. Henceforth unless otherwise specified, static link weights are all equal (and normalized to $1$) and dynamic link weights are always some fixed constant $\mu \in \mathbb Q_+$. Also, there is a single switch and all nodes are connected to it with identical hardware. This is both practically relevant and intuitively realistic; see e.g. Fig. \ref{fig:mirror}. Then the set of switches $S$ of the hybrid network consists of just one switch, which is fully described by the number of switch links each node in the hybrid network has, which we call $\mathrm\Delta_S$. This is closely related to the maximum reconfigurable degree $\mathrm\Delta_R$ from \cite{FPS19}, which is an upper bound on the number of external ports per node. The resulting restriction of RRP can be formalized as follows:

\begin{framed}
    \noindent
    \textbf{\textsc{$\mathrm\Delta_S$-switched RRP ($\sigma$)}}\\
    \emph{Input:} $(G, \mu, D, \kappa)$: $D$ is a workload matrix for the hybrid network $G(S)$ with static (resp. dynamic) links all have weight $1$ (resp. $\mu$) (where $S$ consists of a single switch that every node in $G$ is connected to exactly $\mathrm\Delta_S$ times).\\
    \emph{Question:} Does $G(S)$ admit some configuration $N$ such that the total workload cost of $D$ under $N$ (where the number of alternations for any path is bounded by $\sigma$) is at most $\kappa$?
\end{framed}

\section{Results}\label{sec:results}

Table \ref{table:results} shows a summary of hardness results from previous work as well as our three main intractability results. In general terms, we obtain NP-completeness for \textsc{2-switched RRP} and \textsc{3-switched RRP} on any fixed class of static networks of practical interest (defined more fully below) and for any value of $\sigma$. We then restrict our focus (and associated parameters) to the case where the static network is a hypercube when we establish the NP-completeness of \textsc{1-switched RRP($\sigma=3$)} in this setting; we conjecture that a similar construction can be used to establish hardness when $\sigma>3$. We also, in Theorem \ref{thm:sigma0_poly}, show that \textsc{1-switched RRP($\sigma=0$)} is solvable in polynomial time. 
The cases when $\sigma\in \{1,2\}$ remain interesting open problems. Subsection \ref{sec:lunar} is devoted to the case of \textsc{1-Switched RRP}($\delta=1$) (i.e. any flow-path must contain at most a single dynamic link) which entails at most 2 alternations, and we show this case is NP-complete for hypercubes, grids, and toroidal grids. The proof in that section is intended to facilitate the task of proving hardness for other interesting families of networks.

\def\arraystretch{1.5}
\begin{table}[!t] \footnotesize
		\centering
		\setlength{\tabcolsep}{1pt}
          \begin{tabular}{
				| >{~}m{.18\linewidth}
				| >{\centering\arraybackslash}m{.09\linewidth} 
				| >{\centering\arraybackslash}m{.07\linewidth} 
				| >{\centering\arraybackslash}m{.07\linewidth} 
                | >{\centering\arraybackslash}m{.12\linewidth} 
				| >{\centering\arraybackslash}m{.18\linewidth} 
                | >{\centering\arraybackslash}m{.24\linewidth} 
				  |}
			 
			\hline
			Result  
			& $|S|$
            & $\mathrm\Delta_R$
            & $\sigma$/$\delta$
			& $D$
            & link weights
            & notes
            \\
            \hline
        	\cite{FFSV20}, Theorem 1
			& $\mathrm\Theta(n)$ \hspace{5pt} (or 1 \tablefootnote{By using variable $\mu$ with prohibitively large weights, it is possible to simulate many switches with just one.})
            & $\mathrm\Theta(n)$
            & \multirowcell{3}{\vspace{10pt}\\any \\$\sigma \ge 2$}
			& \multirowcell{3}{\vspace{10pt}\\ sparse,\\ all values\\ $0$ or $1$}
            & variable; $w \in [1,100n^2]$ $\mu \in [1, 100n^2]$    
            & Showed inapprox.  within $\mathrm\Omega(\log n)$
            \\
            \cline{1-3} \cline{6-7}
			\cite{FGS18}, Lemma 1
			& $\mathrm\Theta(n)$
            & \multirowcell{2}{$1$}
            & 
			& 
            & \multirowcell{3}{\vspace{23pt} \\ fixed; \\ $w=\mu=1$}
            & \makecell{All switches \\have 3 ports.}
            \\
            \cline{1-2} \cline{7-7}
			\cite{FGS18}, Theorem 2
			& \multirowcell{5}{\vspace{85pt}\\$1$}
            & 
            & 
			& 
            & 
            & $G$ has $\mathrm\Theta(n)$ components
            \\
            \cline{1-1} \cline{3-5} \cline{7-7}
            \cite{FPS19}, Theorems 4.1, 4.2
			& 
            & \multirowcell{2}{\vspace{13pt}\\$2$}
            & \multirowcell{3}{\vspace{23pt}\\ any\\ $\sigma \ge 0$}
			& dense, values in $\mathrm{poly}(n)$
            & 
            & $G$ is empty; there are no static links 
            \\
            \cline{1-1}\cline{5-7}
            
            Theorem \ref{thm:2switchedRRP}  
			& 
            & 
            & 
			& \multirowcell{3}{\vspace{32pt}\\sparse, \\values in\\ $\mathrm{poly}(n)$} 
            & fixed; \hspace{15pt} $w=1$ \hspace{5pt} $\mu \in \mathrm\Theta(\frac{1}{\mathrm{poly}(n)})$
            & \multirowcell{2}{$G \in \mathcal H$, where $\mathcal H$\\ is any polynomial\\ family of networks\\ (incl. hypercubes, \\grids, cycles).}
            \\
			\cline{1-1} \cline{3-3} \cline{6-6}
            Theorem \ref{thm:3switchedRRP}  
			& 
            & 3
            & 
			& 
            & fixed; \hspace{15pt} $w=1$ \hspace{5pt} $\mu \in \mathrm\Theta(\frac{1}{\mathrm{log}(n)})$
            & 
            \\
			\cline{1-1} \cline{3-4} \cline{6-7}
            Theorem \ref{thm:hypercubeRRPsigma3}  
			& 
            & \multirowcell{2}{\vspace{10pt}\\$1$}
            & $\sigma=3$
			& 
            & \multirowcell{2}{\vspace{-5pt}\\ fixed; \\ $w=1$ \\ any $\mu \in (0,1)$}
            & $G$ is a hypercube
            \\
			\cline{1-1} \cline{4-4} \cline{7-7}
                        Theorem \ref{thm:lunarRRP}  
			& 
            & 
            & $\delta=1$
			& 
            & 
            & $G$ is a hypercube, grid, or toroidal grid
            \\
			\hline
		\end{tabular}\medskip
		\caption{Settings for some pre-existing hardness results for \textsc{RRP}. $|S|$ is the number of switches; $\mathrm\Delta_R$ is the maximum number of external ports per node; $\sigma$ is the segregation parameter and $\delta$ is the dynamic link limit; $D$ is the workload matrix; $n$ denotes the number of nodes in the instance. }\label{table:results}
	\end{table}


As is standard in NP-hardness proofs, we reduce from known NP-complete problems to instances of RRP; the challenge is that, due to the expansive scope of our theorems, we lose several ``degrees of freedom'' which are used for encoding hard instances in, e.g., \cite{FFS19,FFSV20,FGS18}. Specifically, we may not make use of varying static or dynamic link weights to prohibit certain connections, nor encode any features of the input instance in the topology of the hybrid network $G(S)$. For example, in Lemma 1 \cite{FGS18}, many small switches with two feasible configurations each are used to encode a truth assignment, and in Theorem 1 of \cite{FFSV20} ``bad'' links are given weights of order $\mathrm\Theta(n^2)$. Neither of these mechanisms can be leveraged to obtain hardness in our setting; in this sense, our hardness results are strictly stronger and also harder to obtain than those from \cite{FFS19,FFSV20,FGS18}.
We are constrained to choose a \emph{size} for the network $G$, and then to encode the input instance in the demand matrix $D$.  

Our first two results hold for a wide class of graph families, which may be of broader interest for the study of computational hardness in network problems. Rather than allowing arbitrary static networks in instances of RRP, we wish to force any such static network to come from a fixed family of networks where a \emph{family of networks} $\mathcal{H}$ is an infinite sequence of networks $\{H_i: i \geq 0\}$ so that the size $|H_i|$ of any $H_i$ is less than the size of $H_{i+1}$. However, we wish to control the sequence of network sizes. Consequently, we define a \emph{polynomial family of networks} as being a family of networks $\mathcal{H} = \{H_i: i\geq 0\}$ where there exists a polynomial $p_{\mathcal{H}}(x)$ so that $|H_{i+1}| = p_{\mathcal{H}}(|H_i|)$, for each $i\geq 0$\footnote{Technically, we insist that there exists a polynomial Turing machine $\mathcal M$ which computes $H_{i+1}$ on input $H_i$, for each $i \geq 0$, but this definition obfuscates the utility of this description.}. Note that given any $n\geq 0$, we can determine in time polynomial in $n$ the smallest $i$ such that $n\leq |H_i|$. As an example of a polynomial family of networks, consider the hypercubes; here, the polynomial $p_{\mathcal{H}}(x) = 2x$. Other examples include independent sets, complete graphs, cycles, complete binary trees and square grids, among many others. The sweeping generality of having a single construction which holds for any polynomial family $\mathcal H$ poses a challenge in our proofs of Theorems \ref{thm:2switchedRRP} and \ref{thm:3switchedRRP}; we require that our constructed network $H(S)$ behaves identically when $H$ is a connected (or even complete) graph and, at the opposite extreme, when $H$ is disconnected (or even independent).

\begin{theorem} \label{thm:2switchedRRP}
For any polynomial family of networks $\mathcal{H} = \{H_i: i \geq 0\}$, the problem \textsc{2-switched RRP} restricted to instances $(H, \mu, D, \kappa)$ satisfying:
\begin{itemize}
	\item $H\in\mathcal{H}$ has size $n$
	\item the workload matrix $D$ is sparse and all values in it are polynomial in $n$
	\item $\mu \in \mathrm\Theta(\frac{1}{poly(n)})$ is fixed for all dynamic links
\end{itemize}
is NP-complete.
\end{theorem}

\begin{proof}	
Note that \textsc{2-switched RRP} as in the statement of the theorem is in NP as it can be straightforwardly reduced to an equivalent instance of the more general \textsc{RRP}, which is known to be in NP. We now build a polynomial-time reduction from the problem \textsc{3-Min-Bisection} to our restricted version of RRP where the problem \textsc{3-Min-Bisection} is defined as follows:
\begin{itemize}
	\item instance of size $n$: a $3$-regular graph $G=(V,E)$ on $n$ nodes and an integer $k\leq n^2$
	\item yes-instance: there exists a partition of $V$ into two disjoint subsets $A$ and
	$B$, each of size $\frac{n}{2}$, so that the set of edges incident with both
	a node in $A$ and a node in $B$ has size at most $k$; that is, $G$ has \emph{bisection width} at most $k$.
\end{itemize}
Note that any $3$-regular graph necessarily has an even number of nodes. The problem \textsc{3-Min-Bisection} was proven to be NP-complete in \cite{BK02} where it was also shown that approximating \textsc{3-Min-Bisection} to within a constant factor approximation ratio entails the existence of a constant factor approximation algorithm for the more general and widely-studied problem \textsc{Min-Bisection}, which is defined as above but where $G$ can be arbitrary and where $A$ and $B$ have sizes differing by at most one. Given an arbitrary instance $(G=(V,E),k)$ of size $n$ of \textsc{3-Min-Bisection}, we now build an instance $(H, \mu, D, \kappa)$ of \textsc{2-switched RRP}. Moreover, we may assume that $k\leq \frac{n}{3}+46$ as it was proven in \cite{CE88} that every 3-regular graph has bisection width at most $\frac{n}{3}+46$.

First, we define a weighted digraph $D^\prime = (V^\prime, E^\prime)$ which will encode a description of our workload matrix $D$ via: there is a directed edge $(u,v)$ with weight $w$ if, and only if, there is a node-to-node workload of $w$ from $u$ to $v$. Let $\bar{n}$ be the size of the network $H_i$ where $i$ is the smallest integer such that $n+6n^2+2 \leq |H_i|$ and set $H=H_i$.

\begin{itemize}
	\item The node set $V^\prime$ is taken as a disjoint copy of the node set $V$ of $G$, which we also refer to as $V$, together with the set of nodes $V_c = \{x_i,y_i: -\frac{L}{2} \leq i \leq \frac{L}{2}\}$, where $L = 3n^2$ (recall, $n$ is even), and another set of nodes $U$ of size $\bar{n}-(n+6n^2+2)$; so, $|V^\prime|=\bar{n}$. We call every node of $V_c$ a \emph{chain-node}. For ease of presentation, we denote the chain-nodes $x_{\frac{L}{2}}$ and $x_{-\frac{L}{2}}$ by $x^+$ and $x^-$, respectively, and we define the chain-nodes $y^+$ and $y^-$ analogously.
	\item The (directed) edge set $E^\prime$ consists of $E_\alpha\cup E_\beta \cup E_1$ where:
	\begin{itemize}
		\item the set of \emph{chain-edges} $E_\alpha = \{(x_i,x_{i+1}), (y_i,y_{i+1}) : 0 \leq i < \frac{L}{2}\}\cup \{(x_i,x_{i-1}), (y_i,y_{i-1}) : -\frac{L}{2} < i \leq 0\}$
		\item the set of \emph{star-edges} $E_\beta = \{(x_0,v): v \in V\} \cup \{(y_0,v): v \in V\}$
		\item the set of \emph{unit-edges} $E_1$ which is a copy of the edges $E$ of $G$, but on our (copied) node set $V$ and so that every edge is replaced by a directed edge of arbitrary orientation.
	\end{itemize}
\end{itemize}
Note that the nodes of $U$ are all isolated in $D^\prime$ and that $|V^\prime| = \bar{n}$ (the nodes of $U$ will play no role in the following construction). The workloads on the edges of $E^\prime$ are $\alpha$, $\beta$ or $1$ depending upon whether the edge is a chain-edge from $E_\alpha$, a star-edge from $E_\beta$ or a unit-edge from $E_1$, respectively, where we define $\alpha = 24n^6$ and $\beta = 6n^3$. If the directed edge $(u,v)$ has weight $\alpha$ (resp. $\beta$, $1$) in $D^\prime$ then we say that $(u,v)$ is an $\alpha$-demand (resp. $\beta$-demand, $1$-demand). The digraph $D^\prime$ can be visualized as in Fig~\ref{fig:Dprime0}. The grey rectangle denotes the nodes of $V$, the nodes of $V_c$ appear along the top and the bottom and the dashed (resp. dotted, solid) directed edges depict the chain-edges (resp. star-edges, unit-edges). The nodes of $U$ are omitted.

\begin{figure*}
	\centering
	\scalebox{1}{
		\begin{tikzpicture}
			\node[state,circle,scale=0.3,fill=black] (x0) at (0,0) {};
			\node at (0,0.4) {$x_0$};
			\node[state,circle,scale=0.3,fill=black] (xmin1) at (-0.75,0) {};
			\node at (-0.75,0.4) {$x_{-1}$};
			\node[state,circle,scale=0.3,fill=black] (x1) at (0.75,0) {};
			\node at (0.75,0.4) {$x_1$};
			\node[state,circle,scale=0.3,fill=black] (xmin2) at (-1.5,0) {};
			\node at (-1.5,0.4) {$x_{-2}$};
			\node[state,circle,scale=0.3,fill=black] (x2) at (1.5,0) {};
			\node at (1.5,0.4) {$x_2$};
			\node[state,circle,scale=0.3,fill=black] (xmin) at (-4.5,0) {};
			\node at (-4.5,0.4) {$x^{-}=x_{\frac{-L}{2}}$};
			\node[state,circle,scale=0.3,fill=black] (xplus) at (4.5,0) {};
			\node at (4.5,0.4) {$x^{+}=x_{\frac{L}{2}}$};
			\path (x0) edge[->, dashed] (xmin1);
			\path (x0) edge[->, dashed] (x1);
			\path (xmin1) edge[->, dashed] (xmin2);
			\path (x1) edge[->, dashed] (x2);
			\path (xmin2) edge[->, dashed] (-2.25,0);
			\path (x2) edge[->, dashed] (2.25,0);
			\path (-3.75,0) edge[->, dashed] (xmin);
			\path (3.75,0) edge[->, dashed] (xplus);
			\node at (-3,0) {$\ldots$};
			\node at (3,0) {$\ldots$};
			\node[state,circle,scale=0.3,fill=black] (y0) at (0,-5) {};
			\node at (0,-5.4) {$y_0$};
			\node[state,circle,scale=0.3,fill=black] (ymin1) at (-0.75,-5) {};
			\node at (-0.75,-5.4) {$y_{-1}$};
			\node[state,circle,scale=0.3,fill=black] (y1) at (0.75,-5) {};
			\node at (0.75,-5.4) {$y_1$};
			\node[state,circle,scale=0.3,fill=black] (ymin2) at (-1.5,-5) {};
			\node at (-1.5,-5.4) {$y_{-2}$};
			\node[state,circle,scale=0.3,fill=black] (y2) at (1.5,-5) {};
			\node at (1.5,-5.4) {$y_2$};
			\node[state,circle,scale=0.3,fill=black] (ymin) at (-4.5,-5) {};
			\node at (-4.5,-5.4) {$y^{-}=y_{\frac{-L}{2}}$};
			\node[state,circle,scale=0.3,fill=black] (yplus) at (4.5,-5) {};
			\node at (4.5,-5.4) {$y^{+}=y_{\frac{L}{2}}$};
			\path (y0) edge[->, dashed] (ymin1);
			\path (y0) edge[->, dashed] (y1);
			\path (ymin1) edge[->, dashed] (ymin2);
			\path (y1) edge[->, dashed] (y2);
			\path (ymin2) edge[->, dashed] (-2.25,-5);
			\path (y2) edge[->, dashed] (2.25,-5);
			\path (-3.75,-5) edge[->, dashed] (ymin);
			\path (3.75,-5) edge[->, dashed] (yplus);
			\node at (-3,-5) {$\ldots$};
			\node at (3,-5) {$\ldots$};
			\node[state,rectangle,draw=white,scale=1,fill=lightgray, minimum width = 9.5cm, minimum height = 3.5cm] (rect1) at (0,-2.5) {};
			\node[state,circle,scale=0.3,fill=black] (u3) at (4,-1.5) {};
			\node[state,circle,scale=0.3,fill=black] (u5) at (0.75,-1.2) {};
			\node[state,circle,scale=0.3,fill=black] (u0) at (-4,-2.2) {};
			\node[state,circle,scale=0.3,fill=black] (u4) at (-0.5,-2.5) {};
			\node[state,circle,scale=0.3,fill=black] (u1) at (-2,-3) {};
			\node[state,circle,scale=0.3,fill=black] (u6) at (3,-3.2) {};
			\node[state,circle,scale=0.3,fill=black] (u2) at (0.9,-3.5) {};
			\path (x0) edge[->, densely dotted] (u0);
			\path (x0) edge[->, densely dotted] (u1);
			\path (x0) edge[->, densely dotted] (u2);
			\path (x0) edge[->, densely dotted] (u3);
			\path (x0) edge[->, densely dotted] (u4);
			\path (x0) edge[->, densely dotted] (u5);
			\path (x0) edge[->, densely dotted] (u6);
			\path (y0) edge[->, densely dotted] (u0);
			\path (y0) edge[->, densely dotted] (u1);
			\path (y0) edge[->, densely dotted] (u2);
			\path (y0) edge[->, densely dotted] (u3);
			\path (y0) edge[->, densely dotted] (u4);
			\path (y0) edge[->, densely dotted] (u5);
			\path (y0) edge[->, densely dotted] (u6);
			\path (u0) edge[->] (u1);
			\path (u1) edge[->] (u2);
			\path (u2) edge[->] (u3);
			\path (u3) edge[->] (u4);
			\path (u4) edge[->] (u5);
			\path (u6) edge[->] (u4);
			\path (u4) edge[->] (u0);
		\end{tikzpicture}	
	}
	\caption{The graph $D^\prime$.}
	\label{fig:Dprime0}
\end{figure*}

As stated earlier, our static network $H$ is the network $H_i\in \mathcal{H}$ where $|H_i|=\bar{n}$. We refer to the node set of $H$ as $V^\prime$ also and we refer to the subset of nodes within $V^\prime$ corresponding to $V$ as $V$ also. Since we are in the 2-swtiched setting, we have one switch $s$ with $2|V^\prime|$ ports so that every node of $H$ is adjacent, via switch links, to exactly two ports of the switch. Hence, our switch set is $S = \{s\}$ and our hybrid network is $H(S)$. It is important to note that for any configuration $N$, any node of $H(N)$ can be adjacent to at most $2$ other nodes via dynamic links (as $\Delta_S=2$).

As can be seen, we have the graph $G=(V,E)$, the digraph $D^\prime=(V^\prime,E^\prime)$ and the hybrid network $H(S)$ with node set $V^\prime$. Although $G$, $D^\prime$ and $H(S)$ are disjoint in terms of node sets, we do not distinguish between, say, the node set $V$ of $G$ and the subset of nodes $V$ of $H$. It should always be obvious as to which set we are referring to. We proceed similarly when we talk of specific nodes. As ever, we refer to `edges' in graphs and `links' in networks (but they are really one and the same).

We set the weight of any dynamic link as $\mu = \frac{1}{2L} = \frac{1}{6n^2}$ and the bound $\kappa$ for the total workload cost as $\kappa = \kappa_\alpha + \kappa_\beta + \kappa_1$ where:
\begin{itemize}
	\item $\kappa_\alpha = 24n^6$
	\item $\kappa_\beta = 3n^4+\frac{n^3}{2}+n^2$
	\item $\kappa_1 = \frac{k}{2} + \frac{1}{8}  - \frac{1}{4n} + \frac{k}{3n^2}.$
\end{itemize}
The values of $\kappa_\alpha$, $\kappa_\beta$ and $\kappa_1$ have the following significance. 
\begin{itemize}
	\item Suppose that for every chain-edge $(x_i,x_{i+1})$ (resp. $(x_i,x_{i-1})$, $(y_i,y_{i+1})$, $(y_i,$\linebreak$y_{i-1})$) of $E_\alpha$, we force a dynamic link joining $x_i$ and $x_{i+1}$ (resp. $x_i$ and $x_{i-1}$, $y_i$ and $y_{i+1}$, $y_i$ and $y_{i-1}$) in $H(N)$ and choose a corresponding flow-path serving this $\alpha$-demand as consisting of this dynamic link ($N$ is the resulting configuration from our chosen switch matching). The total workload cost of flow-paths serving $\alpha$-demands is $2L\alpha\mu = 24n^6 = \kappa_\alpha$.
	\item Further, suppose that the dynamic links incident with nodes of $V$ in $H(N)$ are chosen so that we have a path of dynamic links $p_A$ from $x^+$ to either $y^-$ or $y^+$, involving the subset of nodes $A\subseteq V$, and a path of dynamic links $p_B$ from $x^-$ to $y^+$ or $y^-$, respectively, involving the subset of nodes $B\subseteq V$, so that both $p_A$ and $p_B$ have length $\frac{n}{2}+1$. That is, we choose the dynamic links so that they form a cycle $C$ (of length $n+2L+2$) in $H(N)$ covering exactly the nodes of $V$ and $V_c$. Suppose that for any star-edge $(x_0,v)$ (resp. $(y_0,v)$) of $E_\beta$, we choose the flow-path in $H(N)$ serving this star-edge as consisting entirely of dynamic links resulting from the shortest path in our cycle $C$ from $x_0$ to $v$ (resp. $y_0$ to $v$). The total workload cost of flow-paths corresponding to the star-edges is
	$$4\mu\beta\sum_{i=1}^{\frac{n}{2}}(\frac{L}{2} + i) = 3n^4+\frac{n^3}{2}+n^2 = \kappa_\beta.$$
	\item Further, suppose we choose the flow-path in $H(N)$ serving the $1$-demand $(u,v)$ (in $E_1$) to be a path of dynamic links within the cycle $C$ of shortest length. If $u$ and $v$ both lie on $p_A$ or both lie on $p_B$ then the workload cost of this flow-path is at most $\mu(\frac{n}{2}-1) = \frac{1}{6n}(\frac{1}{2}-\frac{1}{n})$, and if one of $u$ and $v$ lies on $p_A$ with the other node lying on $p_B$ then the workload cost of this flow-path is at most $\mu(\frac{n}{2}+L+1) = \frac{1}{2} + \frac{1}{12n} + \frac{1}{6n^2}$. If the width of the bisection of $G$ formed by $A$ and $B$ is at most $k$ then the total workload cost of flow-paths corresponding to the unit-edges is at most
	$$(\frac{3n}{2} - k)\mu(\frac{n}{2}-1) + k\mu(\frac{n}{2}+L+1) = \frac{k}{2}+\frac{1}{8} - \frac{1}{4n} + \frac{k}{3n^2}  = \kappa_1.$$
\end{itemize}
From above, we immediately obtain that if $(G,k)$ is a yes-instance of \textsc{3-Min-Bisection} then $(H(S),D,\mu,\kappa)$ is a yes-instance of RRP.

Conversely, suppose that $(H,\mu, D,\kappa)$ is a yes-instance of RRP. Let $N$ be a configuration (consisting of a switch matching of $s$) and let $F$ be a collection of flow-paths that witness that the total workload cost for the workload matrix $D$ is at most $\kappa$. Denote by $N$, also, the sub-network of $H(N)$ consisting solely of dynamic links. W.l.o.g. we may assume that every node of $H(N)$ is incident with exactly two dynamic links of $N$ (by adding additional dynamic links that we do not actually use in any flow-path, if necessary).

\begin{claim}\label{clm:2switchedRRP1}
If $(u,v)\in E_\alpha$ (that is, $(u,v)$ is a chain-edge in $D^\prime$) then $(u,v)$ is a dynamic link in $N$; so, the total workload cost of the flow-paths serving $\alpha$-demands is exactly $\kappa_\alpha = 24n^6$.
\end{claim}

\begin{proof}
Suppose that $(x_i,x_{i+1})\not\in N$ (there is an analogous argument for each of $(x_i,x_{i-1})$, $(y_i,y_{i+1})$ and $(y_i,y_{i-1})$). So, the flow-path of $F$ serving the $\alpha$-demand $(x_i,x_{i+1})$ consists of at least two links and has workload cost at least $2\mu\alpha$ (as $2\mu$ is strictly less than $1$ which is the weight of a static link). Hence, the total workload cost of flow-paths corresponding to the chain-edges is at least $\kappa_\alpha + \mu\alpha$. Denote the total workload cost of flow-paths serving the $\beta$-demands (resp. $1$-demands) by $\bar{\kappa}_\beta$ (resp. $\bar{\kappa}_1$). So, $\kappa_\alpha + \mu\alpha + \bar{\kappa}_\beta + \bar{\kappa}_1 \leq \kappa$ with $\mu\alpha \leq \mu\alpha + \bar{\kappa}_\beta + \bar{\kappa}_1 \leq \kappa_\beta + \kappa_1$; that is, with $n^4 \leq \frac{n^3}{2}+n^2 + \frac{k}{2} + \frac{1}{8} -\frac{1}{4n} + \frac{k}{3n^2}$. This yields a contradiction (so long as $n$ is big enough) and the claim follows.
\qed \end{proof}

\begin{claim}\label{clm:2switchedRRP2}
The set of dynamic links $N$ forms a cycle in $H(N)$ covering exactly the nodes of $V\cup V_c$ and on which every link from $\{(x_i,x_{i+1}), (y_i,y_{i+1}) : 0 \leq i < \frac{L}{2}\}\cup \{(x_i,x_{i-1}), (y_i,y_{i-1}) : -\frac{L}{2} < i \leq 0\}$ lies. Moreover, every flow-path of $F$ serving a $\beta$-demand consists entirely of dynamic links.
\end{claim}

\begin{proof}
By Claim~\ref{clm:2switchedRRP1}, if $(u,v)\in E_\alpha$ then $(u,v)\in N$ and the total workload cost of the flow-paths serving the $\alpha$-demands is exactly $\kappa_\alpha = 24n^6$. Note that because every node of $H(N)$ is adjacent to at most $2$ dynamic links, there are no other dynamic links incident with a node from $V_c\setminus\{x^+,x^-,y^+,y^-\}$. 

Suppose that a flow-path of $F$ serving a $\beta$-demand consists entirely of dynamic links. So, the workload cost of such a flow-path is at least $\mu\beta(\frac{L}{2}+1) = \frac{3n^3}{2}+n$.
Alternatively, if a flow-path of $F$ serving a $\beta$-demand contains at least one static link then the workload cost of such a flow-path is at least $\beta = 6n^3$. 
Consequently, if at least one flow-path of $F$ serving a $\beta$-demand contains a static link (of weight $1$) then the total workload cost of flow-paths serving the $\beta$-demands is at least $(2n-1)\mu\beta(\frac{L}{2}+1) + \beta = (2n-1)(\frac{3n^3}{2}+n) + 6n^3 = 3n^4 +\frac{9n^3}{2}+2n^2 - n > 3n^4 + \frac{n^3}{2} + n^2 + \frac{k}{2} + \frac{1}{8}  - \frac{1}{4n} + \frac{k}{3n^2} = \kappa_\beta +\kappa_1 = \kappa - \kappa_\alpha$ which yields a contradiction. Thus, all flow-paths of $F$ serving a $\beta$-demand consist entirely of dynamic links. In particular, $N$ must be connected. As $N$ is regular of degree $2$ then $N$ must, in fact, be a cycle covering the nodes of $V\cup V_c$. The claim follows.\qed \end{proof}

By Claim~\ref{clm:2switchedRRP2}, $N$ consists of a path of dynamic links involving all links from $\{(x_i,x_{i+1}) : -\frac{L}{2} \leq i < \frac{L}{2}\}$ from $x^-$ to $x^+$ concatenated with a path $p_A$ of dynamic links from $x^+$ to either $y^-$ or $y^+$ concatenated with a path of dynamic links involving all links from $\{(y_i,y_{i+1}) : -\frac{L}{2} \leq i < \frac{L}{2}\}$ from $y^-$ or $y^+$ to $y^+$ or $y^-$ concatenated with a path $p_B$ of dynamic links from $y^+$ or $y^-$ to $x^-$, respectively. W.l.o.g. we may assume that $p_A$ runs from $x^+$ to $y^+$ and $p_B$ from $y^-$ to $x^-$.

We can now recalculate the total workload cost of the flow-paths serving the $\beta$-demands. Suppose that there are $p$ nodes of $V$ on $p_A$ and so $n-p$ nodes of $V$ on $p_B$. We may assume that no nodes of $U$ lie on $p_A$ or $p_B$: if any do, then removing all such nodes from those paths results in a strictly smaller total workload cost. The total workload cost of the flow-paths serving the $\beta$-demands is
\begin{eqnarray*}
	\lefteqn{2\mu\beta\sum_{i=1}^{p}(\frac{L}{2}+i) + 2\mu\beta\sum_{i=1}^{n-p}(\frac{L}{2}+i)}\\
	& = & p\mu\beta L + (n-p)\mu\beta L + 2\mu\beta\sum_{i=1}^{p}i + 2\mu\beta\sum_{i=1}^{n-p}i\\
	& = & n\mu\beta L + (p(p+1) + (n-p)(n-p+1))\mu\beta\\
	& = & 3n^4 + n(2p^2 -2np +n^2 + n)\\
	& \geq & 3n^4 + \frac{n^3}{2} + n^2 = \kappa_\beta \mbox{ (when $p$ takes the value $\frac{n}{2}$)}
\end{eqnarray*}
as the minimum value for $3n^4 + n(2p^2 -2np +n^2 + n)$ is when $p=\frac{n}{2}$.

Suppose that the paths $p_A$ and $p_B$ have different lengths. From above, the total workload cost of the flow-paths serving the $\beta$-demands is at least $3n^4 + n(2(\frac{n}{2} + 1)^2 -2n(\frac{n}{2} + 1) +n^2 + n) = 3n^4 + \frac{n^3}{2} + n^2 + 2n  = \kappa_\beta + 2n$. Consequently, we must have that $2n \leq \kappa_1 = \frac{k}{2} + \frac{1}{8}  - \frac{1}{4n} + \frac{k}{3n^2}$. As $k\leq \frac{n}{3}+46$ (from \cite{CE88}), we have that $2n\leq \frac{n}{6} + 23 + \frac{1}{8} - \frac{1}{4n} + \frac{1}{9n} + \frac{46}{3n^2}$ which yields a contradiction. Hence, the paths $p_A$ and $p_B$ each have length $\frac{n}{2}$, with the total workload cost of the flow-paths corresponding to the unit-edges being at most $\kappa_1$.

Let $(u,v)$ be a unit-edge. From above: if $u$ and $v$ both lie on $p_A$ or both lie on $p_B$ and the flow-path serving the $1$-demand $(u,v)$ consists entirely of dynamic links then the workload cost of this flow-path lies between $\mu$ and $\mu(\frac{n}{2}-1)$, i.e., $\frac{1}{6n^2}$ and $\frac{1}{6n}(\frac{1}{2} - \frac{1}{n})$; if $u$ and $v$ lie on $p_A$ and $p_B$, respectively, or vice versa, and the flow-path serving the $1$-demand $(u,v)$ consists entirely of dynamic links then the workload cost of this flow-path lies between $\mu(L+2)$ and $\mu(\frac{n}{2}+L+1)$, i.e., $\frac{1}{2} + \frac{1}{3n^2}$ and $\frac{1}{2}+\frac{1}{12n} +\frac{1}{6n^2}$; and if the flow-path serving the $1$-demand $(u,v)$ contains a static link then the workload cost of this flow-path is at least 1. In particular, we may assume that any flow-path corresponding to a unit-edge consists entirely of dynamic links.

Let $A$ (resp. $B$) be the nodes of $V$ that appear on $p_A$ (resp. $p_B$). So, $(A,B)$ is a bisection of $G$. Suppose that this bisection has width $k+\epsilon$, for some $\epsilon \geq 1$. So there are $k+\epsilon$ unit-edges whose corresponding work-flows contribute collectively at least $(k+\epsilon)\mu(L+2) \geq (k+1)(\frac{1}{2} + \frac{1}{3n^2}) = \frac{k}{2} + \frac{1}{2} + \frac{1}{3n^2} + \frac{k}{3n^2} > \frac{k}{2} + \frac{1}{8}  - \frac{1}{4n} + \frac{k}{3n^2} = \kappa_1$ to the total workload cost, which yields a contradiction. Consequently, $G$ has bisection width at most $k$ and we have that if $(H,\mu,D,\kappa)$ is a yes-instance of RRP then $(G,k)$ is a yes-instance of \textsc{3-Min-Bisection}. Our result follows as $(H,\mu,D,\kappa)$ can be constructed from $(G,k)$ in time polynomial in $n$.
\qed \end{proof}

This result significantly strengthens Theorems 4.1 and 4.2 from \cite{FPS19}: there, RRP($\mathrm\Delta_R \geq 2, \sigma=0$) is shown to be NP-complete when the static network is an independent set, and the proof does not enable us to restrict the workload matrix $D$ meaningfully. The main weakness of Theorem~\ref{thm:2switchedRRP} is its reliance on $\mu$ being a polynomial factor smaller than any static link weight. This is actually related to the fact that a connected 2-regular network, as is $G(N)$ when $G$ is an independent set and $\mathrm\Delta_S=2$, has diameter linear in the number of nodes $n$. 
A network of maximum degree 3, on the other hand, may have diameter logarithmic in $n$ (e.g., a complete binary tree has this property) and we indeed show NP-completeness of RRP($\mathrm\Delta_S=3$) when $\mu= \mathrm\Theta(\frac{1}{\log n})$.

\begin{theorem}\label{thm:3switchedRRP}
For any polynomial family of networks $\mathcal{H} = \{H_i: i \geq 0\}$, the problem \textsc{3-switched RRP} restricted to instances $(H, \mu, D, \kappa)$ satisfying:
\begin{itemize}
	\item $H\in\mathcal{H}$ has size $n$
	\item the workload matrix $D$ is sparse and all values in it are polynomial in $n$
	\item $\mu \in \mathrm\Theta(\frac{1}{\log(n)})$
\end{itemize}
is NP-complete.
\end{theorem}

\begin{proof}
\setcounter{claim}{0}
As in the case of Theorem \ref{thm:2switchedRRP}, membership of NP is straightforward. The problem \textsc{Restricted Exact Cover by 3-Sets (RXC3)} is defined as follows:
\begin{itemize}
	\item instance of size $n$: a finite set of \emph{elements} $\mathcal{X} = \{x_i: 1\leq i \leq 3n\}$, for some $n\geq 1$, and a collection $\mathcal{C}$ of $3$-element subsets of $\mathcal{X}$, called \emph{clauses}, where $\mathcal{C}=\{c_j: 1\leq j\leq 3n\}$ and where each $c_j=\{x^j_1,x^j_2,x^j_3\}$ so that every element of $\mathcal{X}$ appears in exactly three clauses of $\mathcal{C}$
	\item yes-instance: $\mathcal{C}$ contains an \emph{exact cover} $\mathcal{C}^\prime$ for $\mathcal{X}$; that is, a subset $\mathcal{C}^\prime \subseteq \mathcal{C}$ so that every element of $\mathcal{X}$ appears in exactly one clause of $\mathcal{C}^\prime$ (hence, $\mathcal{C}^\prime$ necessarily has size $n$).
\end{itemize}
The problem \textsc{RXC3} is NP-complete as was proven in Theorem~A.1 of \cite{Gon85}. Let $(\mathcal{X},\mathcal{C})$ be an instance of \textsc{RXC3} of size $n$. We will build an instance $(H, \mu,D,  \kappa)$ of \textsc{3-switched RRP} corresponding to $(\mathcal{X},\mathcal{C})$.
	
First, we define a weighted graph $D^\prime$ so as to describe the workload $D$. Having built $D^\prime$, we will orient every undirected edge so that: there is a directed edge $(u,v)$ with weight $w$ if, and only if, there is a node-to-node workload of $w$ from $u$ to $v$ (we still refer to the resulting digraph as $D^\prime$). Consequently, we will need one flow-path in our resulting hybrid network corresponding to each directed edge of $D^\prime = (V^\prime,E^\prime)$. Our construction proceeds in stages.
\begin{itemize}
	\item We begin with a complete binary tree $T$ of depth $d$ so that the number of leaves is $2^d$ with $2^{d-1} < n \leq 2^d$. If the leaves are $\{l_i: 1\leq i \leq 2^d\}$ then we colour: every leaf $l_i$, where $n<i$, white; every leaf $l_i$, where $1\leq i\leq n$, grey; and the root of $T$ grey. Throughout, if we do not specify the colour of a node as grey or white then it is coloured black. Note that $d=\lceil \log_2(n)\rceil$.
	\item For every $x_i\in \mathcal{X}$, we create a unique node $x_i$ and for every clause $c_j\in \mathcal{C}$, we create a unique $2$-path as follows: there are $3$ nodes $u_j$, $v_j$ and $w_j$ together with the edges $(u_j,v_j)$ and $(v_j,w_j)$ and the node $u_j$ is coloured grey.
	\item For each $1\leq j\leq 3n$, if the clause $c_j=\{x^j_1,x^j_2,x^j_3\}$ then there are edges $(x_1^j,u_j)$, $(x^j_2,w_j)$ and $(x_3^j, w_j)$.
	\item To every node coloured grey, we attach a unique gadget consisting of a $4$-clique with one of its edges subdivided by a node, denoted $z$, so that there is an edge joining the gray node and the node $z$; and to every node coloured white, we attach the same gadget except that we identify the white node and the node denoted $z$. These attachments can be visualized as in Fig.~\ref{fig:gadgetb} and Fig.~\ref{fig:gadgeta} where the white and gray nodes are as shown.
	\item There is an edge $(r,x_i)$, for all $1\leq i\leq 3n$: call these the \emph{root-edges}.
\end{itemize}
	
\begin{figure*}
	\centering
	\scalebox{1}{
			\subfloat[\label{fig:gadgetb}]{
			\begin{tikzpicture}
				\node[state,circle,scale=0.3,fill=white] (r) at (0,0) {};
				\node[state,circle,scale=0.3,fill=black] (a) at (2,1) {};
				\node[state,circle,scale=0.3,fill=black] (d) at (2,-1) {};
				\node[state,circle,scale=0.3,fill=black] (b) at (3,1) {};
				\node[state,circle,scale=0.3,fill=black] (c) at (3,-1) {};
				\path (r) edge[-] (a);
				\path (r) edge[-] (d);
				\path (a) edge[-] (b);
				\path (a) edge[-] (c);
				\path (d) edge[-] (b);
				\path (d) edge[-] (c);
				\path (b) edge[-] (c);
			\end{tikzpicture}
		}
	\hspace{0.5in}
\subfloat[\label{fig:gadgeta}]{
	\begin{tikzpicture}
		\node[state,circle,scale=0.3,fill=lightgray] (r) at (0,0) {};
		\node[state,circle,scale=0.3,fill=black] (z) at (2,0) {};
		\node at (1.6,-0.4) {$z$};
		\node[state,circle,scale=0.3,fill=black] (a) at (3,1) {};
		\node[state,circle,scale=0.3,fill=black] (d) at (3,-1) {};
		\node[state,circle,scale=0.3,fill=black] (b) at (4,1) {};
		\node[state,circle,scale=0.3,fill=black] (c) at (4,-1) {};
		\path (r) edge[-] (z);
		\path (z) edge[-] (a);
		\path (z) edge[-] (d);
		\path (a) edge[-] (b);
		\path (a) edge[-] (c);
		\path (d) edge[-] (b);
		\path (d) edge[-] (c);
		\path (b) edge[-] (c);
	\end{tikzpicture}
}

}
	\caption{Attaching the gadgets.}
	\label{fig:gadget}
\end{figure*}

This completes the construction of the graph $D^\prime$ which can be visualized as in Fig.~\ref{fig:Dprime} where the root-edges are depicted as dashed. Note that the number of nodes in $V^\prime$ is $N = 6(2^d) + 28n + 4$ and the number of edges in $E^\prime$ is $9(2^d)+43n+6$, with $3n$ of these edges root-edges. Let $D^{\prime\prime}$ be $D^\prime$ with the root-edges removed. Note that every node of $D^{\prime\prime}$ has degree $3$ except for the leaf nodes of $\{l_i: 1\leq i \leq n\}$ and the nodes of $\{v_j: 1\leq j \leq m\}$, all of which have degree $2$ (we return to this comment presently). Let $E^\prime = E^{\prime\prime}\cup E^{root}$ where $E^{root}$ is the set of root-edges (and so $E^{\prime\prime}$ is the edge set of $D^{\prime\prime}$). As regards the edge-weights, orient each of the edges of $E^\prime$ `down' the graph $D^\prime$ as it is portrayed in Fig.~\ref{fig:Dprime} so as to obtain a digraph and: give each directed edge of $E^{\prime\prime}$ weight $\alpha$, where $\alpha = 4n\log_2(n)$; and give each directed root-edge of the form $(r,x_i)$ weight $1$.

\begin{figure*}
	\centering
	\scalebox{1}{
		\begin{tikzpicture}
				\node[state,circle,scale=0.3,fill=lightgray] (r) at (0,0) {};
				\node at (0,0.4) {$root$};
				\node[state,circle,scale=0.3,fill=black] (c11) at (-0.75,-0.5) {};
				\node[state,circle,scale=0.3,fill=black] (c12) at (0.75,-0.5) {};
				\node[state,circle,scale=0.3,fill=black] (c21) at (-1.25,-1) {};
				\node[state,circle,scale=0.3,fill=black] (c22) at (-0.25,-1) {};
				\node[state,circle,scale=0.3,fill=black] (c23) at (0.25,-1) {};
				\node[state,circle,scale=0.3,fill=black] (c24) at (1.25,-1) {};
				\node[state,circle,scale=0.3,fill=black] (c31) at (-2.5,-2) {};
				\node[state,circle,scale=0.3,fill=black] (c32) at (-1.25,-2) {};
				\node[state,circle,scale=0.3,fill=black] (c33) at (0.25,-2) {};
				\node[state,circle,scale=0.3,fill=black] (c34) at (2.5,-2) {};
				\node[state,circle,scale=0.3,fill=lightgray] (c41) at (-2.75,-2.5) {};
				\node[state,circle,scale=0.3,fill=lightgray] (c42) at (-2.25,-2.5) {};
				\node[state,circle,scale=0.3,fill=lightgray] (c43) at (-1.5,-2.5) {};
				\node[state,circle,scale=0.3,fill=lightgray] (c44) at (-1,-2.5) {};
				\node[state,circle,scale=0.3,fill=lightgray] (c45) at (0,-2.5) {};
				\node[state,circle,scale=0.3,fill=white] (c46) at (0.5,-2.5) {};
				\node[state,circle,scale=0.3,fill=white] (c47) at (2.25,-2.5) {};
				\node[state,circle,scale=0.3,fill=white] (c48) at (2.75,-2.5) {};
				\node[state,circle,scale=0.3,fill=black] (v1) at (-3.25,-4) {};
				\node[state,circle,scale=0.3,fill=black] (v2) at (-1.75,-4) {};
				\node[state,circle,scale=0.3,fill=black] (v3) at (0.75,-4) {};
				\node[state,circle,scale=0.3,fill=black] (v4) at (3.25,-4) {};
				\node[state,circle,scale=0.3,fill=lightgray] (u1) at (-3.65,-4.5) {};
				\node[state,circle,scale=0.3,fill=black] (w1) at (-2.85,-4.5) {};
				\node[state,circle,scale=0.3,fill=lightgray] (u2) at (-2.15,-4.5) {};
				\node[state,circle,scale=0.3,fill=black] (w2) at (-1.35,-4.5) {};
				\node[state,circle,scale=0.3,fill=lightgray] (u3) at (0.35,-4.5) {};
				\node[state,circle,scale=0.3,fill=black] (w3) at (1.15,-4.5) {};
				\node[state,circle,scale=0.3,fill=lightgray] (u4) at (2.85,-4.5) {};
				\node[state,circle,scale=0.3,fill=black] (w4) at (3.65,-4.5) {};
				\node[state,circle,scale=0.3,fill=black] (x1) at (-3.25,-6) {};
				\node[state,circle,scale=0.3,fill=black] (x2) at (-1.75,-6) {};
				\node[state,circle,scale=0.3,fill=black] (x3) at (-0.25,-6) {};
				\node[state,circle,scale=0.3,fill=black] (x4) at (3.25,-6) {};
				\path (r) edge[-] (c11);
				\path (r) edge[-] (c12);
				\path (c11) edge[-] (c21);
				\path (c11) edge[-] (c22);
				\path (c12) edge[-] (c23);
				\path (c12) edge[-] (c24);
				\path (c31) edge[-] (c41);
				\path (c31) edge[-] (c42);
				\path (c32) edge[-] (c43);
				\path (c32) edge[-] (c44);
				\path (c33) edge[-] (c45);
				\path (c33) edge[-] (c46);
				\path (c34) edge[-] (c47);
				\path (c34) edge[-] (c48);
				\node at (-1.5,-1.5) {$\ldots$};
				\node at (-0.5,-2.25) {$\ldots$};
				\node at (1.38,-2.25) {$\ldots$};
				\node at (-0.5,-4.25) {$\ldots$};
				\node at (2,-4.25) {$\ldots$};
				\node at (1.5,-1.5) {$\ldots$};
				\node at (-1,-6) {$\ldots$};
				\node at (1.5,-6) {$\ldots$};
				\path (u1) edge[-] (v1);
				\path (v1) edge[-] (w1);
				\path (u2) edge[-] (v2);
				\path (v2) edge[-] (w2);
				\path (u3) edge[-] (v3);
				\path (v3) edge[-] (w3);
				\path (u4) edge[-] (v4);
				\path (v4) edge[-] (w4);
				\node at (-2.75,-3) {$l_1$};
				\node at (-2.25,-3) {$l_2$};
				\node at (-1.5,-3) {$l_3$};
				\node at (-1,-3) {$l_4$};
				\node at (0,-3) {$l_n$};
				\node at (0.75,-3) {$l_{n+1}$};
				\node at (2.1,-3) {$l_{2^d-1}$};
				\node at (3.15,-3) {$l_{2^d}$};
				\node at (-3.25,-3.6) {$v_1$};
				\node at (-1.75,-3.6) {$v_2$};
				\node at (0.75,-3.6) {$v_j$};
				\node at (3.25,-3.6) {$v_{3n}$};
				\node at (-3.7,-4.1) {$u_1$};
				\node at (-2.8,-4.1) {$w_1$};
				\node at (-2.2,-4.1) {$u_2$};
				\node at (-1.3,-4.1) {$w_2$};
				\node at (0.3,-4.1) {$u_3$};
				\node at (1.2,-4.1) {$w_3$};
				\node at (2.8,-4.1) {$u_{3n}$};
				\node at (3.7,-4.1) {$w_{3n}$};
				\node at (-3.25,-6.5) {$x_1$};
				\node at (-1.75,-6.5) {$x_2$};
				\node at (-0.25,-6.5) {$x_i$};
				\node at (3.25,-6.5) {$x_{3n}$};
				\draw[dashed, -] (x1) to[out=180,in=180] (r);
				\path (-2.3,-6.2) edge[dashed, -] (x2);
				\path (-0.8,-6.2) edge[dashed, -] (x3);
				\path (r) edge[dashed, -] (-0.6,0.15);
				\path (r) edge[dashed, -] (0.6,0.15);
				\draw[dashed, -] (x4) to[out=0,in=0] (r);
				\draw[<->] (4.5,0) to (4.5,-2.5);
				\node at (4.8,-1.25) {$T$};
				\path (x1) edge[-] (-3.25,-5.5);
				\path (x1) edge[-] (-3.1,-5.5);
				\path (x1) edge[-] (-3.4,-5.5);
				\path (x2) edge[-] (-1.75,-5.5);
				\path (x2) edge[-] (-1.6,-5.5);
				\path (x2) edge[-] (-1.9,-5.5);
				\path (x3) edge[-] (-0.25,-5.5);
				\path (x3) edge[-] (-0.1,-5.5);
				\path (x3) edge[-] (-0.4,-5.5);
				\path (x4) edge[-] (3.25,-5.5);
				\path (x4) edge[-] (3.1,-5.5);
				\path (x4) edge[-] (3.4,-5.5);
				\path (u1) edge[-] (-3.65,-5);
				\path (w1) edge[-] (-3,-5);
				\path (w1) edge[-] (-2.7,-5);
				\path (u2) edge[-] (-2.15,-5);
				\path (w2) edge[-] (-1.2,-5);
				\path (w2) edge[-] (-1.5,-5);
				\path (u3) edge[-] (0.35,-5);
				\path (w3) edge[-] (1.0,-5);
				\path (w3) edge[-] (1.3,-5);
				\path (u4) edge[-] (2.85,-5);
				\path (w4) edge[-] (3.8,-5);
				\path (w4) edge[-] (3.5,-5);
			\end{tikzpicture}
		}
		
	\caption{The graph $D^\prime$.}
	\label{fig:Dprime}
\end{figure*}

Let $\bar{n}$ be the size of the network $H_i$ where $i$ is the smallest integer such that $N \leq |H_i|$ and we fix our static network as $H=H_i$. We name a subset of nodes of $H$ as $V^\prime$ with the remaining nodes named $U$. $\mathrm \Delta_S=3$, so we have one switch $s$ that has $3n$ switch ports and every node of $H$ has $3$ external ports; so, every node of $H$ has exactly $3$ switch links to $s$, and $S = \{s\}$. Finally, we define $\mu=\frac{1}{2\lceil \log_2(n)\rceil}$ and $\kappa=\mu\alpha|E^{\prime\prime}| + 3\mu n(\lceil \log_2(n)\rceil + 3)$ so as to complete the construction of our instance $(H,\mu,D,\kappa)$ of \textsc{3-switched RRP}.

Suppose that $(\mathcal{X},\mathcal{C})$ is a yes-instance of \textsc{RXC3}. So, let $\mathcal{C}^\prime \subseteq \mathcal{C}$ be an exact cover of $\mathcal{X}$. For every $(u,v)\in E^{\prime\prime}$, we choose $(u,v)$ to be a dynamic link and ensure that every node of $\{v_j: 1\leq j\leq 3n, c_j\in \mathcal{C}^{\prime}\}$ is joined via a dynamic link to a leaf node of $\{l_i: 1\leq i\leq n\}$. If $v_j$ is such that $c_j\in \mathcal{C}^\prime$ then denote the unique leaf node to which there is a dynamic link from $v_j$ by $\bar{l}_j$. Denote the resulting switch matching of $s$ as constituting the configuration $N$. Referring back to our comment on node degrees above, the dynamic links so added must form a matching in $H(N)$. For any workload of weight $\alpha$ in $D$, we choose the corresponding flow-path to consist of the dynamic link joining the two nodes in question. Consequently, for all of the workloads corresponding to directed edges of $E^{\prime\prime}$, the total workload cost of the corresponding flow-paths is $\mu\alpha|E^{\prime\prime}|$. All of the remaining workloads involve the root and the nodes of $\{x_i: 1 \leq i \leq n\}$. Any one of these nodes $x_i$ is such that the element $x_i$ in a unique clause $c_j\in \mathcal{C}^\prime$. We choose the flow-path consisting entirely of dynamic links that routes from the root down the tree $T$ to the leaf-node $\bar{l}_j$ and on to $v_j$ and, through either $u_j$ or $w_j$, on to $x_i$. This flow-path has (graph-theoretic) length $d+3$ and consequently the total workload cost due to flow-paths corresponding to edges of $E^{root}$ is exactly $3\mu n(d+3) = 3\mu n(\lceil \log_2(n)\rceil + 3)$. Hence, the total workload cost due to all flow-paths is $\mu\alpha|E^{\prime\prime}| + 3\mu n(\lceil \log_2(n)\rceil + 3)=\kappa$ and $(H,\mu,D,\kappa)$ is a yes-instance of \textsc{3-switched RRP}.

Conversely, suppose that $(H,\mu,D,\kappa)$ is a yes-instance of \textsc{3-switched RRP} and that $N$ is a configuration and $F$ a set of flow-paths witnessing that the total workload cost is at most $\kappa$.

\begin{claim}\label{clm:3switched1}
It is necessarily the case that every directed edge $(u,v)\in E^{\prime\prime}$ is such that $(u,v)$ is a dynamic link in $H(N)$.
\end{claim}
	
\begin{proof}
Suppose, for contradiction, that there is a directed edge $(u,v)\in E^{\prime\prime}$ so that $(u,v)$ is not a dynamic link in $N$. By construction, the node-to-node workload of $\alpha$ from $u$ to $v$ contributes a cost of at least $2\mu\alpha$ to the total workload cost (note that the cost of a static link is $1$ and $2\mu < 1$ when $n\geq 3$). So, the total workload cost from flow-paths corresponding to directed edges of $E^{\prime\prime}$ is at least $\mu\alpha(|E^{\prime\prime}| + 1)$. However, $\kappa=\mu\alpha|E^{\prime\prime}| + 3\mu n(\lceil \log_2(n)\rceil + 3)$ and so we must have that $\mu\alpha = 4\mu n\log_2(n) \leq 3\mu n(\lceil \log_2(n)\rceil + 3)$ which yields a contradiction when $n\geq 2^{10}$. The claim follows.
\qed \end{proof}
	
By Claim~\ref{clm:3switched1}, the total workload cost of flow-paths corresponding to the directed edges of $E^{\prime\prime}$ is exactly $\mu\alpha|E^{\prime\prime}|$. Consequently, the total workload cost of flow-paths corresponding to the directed edges of $E^{root}$ must be at most $3\mu n(\lceil \log_2(n)\rceil + 3)$.

Given the set of dynamic links as supplied by Claim~\ref{clm:3switched1}, the only other possible dynamic links we might have involve the leaf nodes $\{l_i: 1\leq i \leq n\}$, the nodes of $\{v_i: 1\leq i\leq 3n\}$ and nodes of $U$. If we are looking for a path of dynamic links joining $r$ and any node $x_i$ then no matter how we might extend the set of dynamic links corresponding to the directed edges of $E^{\prime\prime}$, we can never obtain a path of (graph-theoretic) length less than $d+3$: a path of length $d$ down the tree $T$ from $r$ to a leaf $l$ followed by a path of length $3$ of the form $(l,v_j,u_j,x_i)$ or $(l,v_j,w_j,x_i)$. 

We may assume $U$ is not involved in any such path. Any path visiting some node $\bar{u}\in U$ can be shortened to a path of the former form by simply removing $\bar{u}$ from the path. This corresponds to a rewiring of the configuration $N$ which decreases the total workload cost
so, for example, the path $(r, \ldots, l,\bar{u},v_j,u_j,x_i)$  of length $n+4$ can be shortened to $(r, \ldots, l,v_j,u_j,x_i)$ of length $n+3$ by removing the dynamic links $(l, \bar{u})$ and $(\bar{u}, v_j)$ from $N$ and adding the dynamic link $(l, v_j)$ to $N$. Clearly this does not increase the workload cost for any demand.

So, any flow-path corresponding to some directed edge $(r,x_i)\in E^{root}$ and consisting entirely of dynamic links has workload cost at least $\mu(d+3)  =\mu(\lceil \log_2(n)\rceil+3)$. If a flow-path corresponding to some directed edge $(r,x_i)\in E^{root}$ contains a static link then the workload cost of this flow-path is at least $1$ and $1 > \frac{1}{2} + \frac{3}{2\lceil \log_2(n)\rceil} = \mu(\lceil \log_2(n)\rceil+3)$, when $n \geq 2^4$. As there are $3n$ flow-paths corresponding to directed edges of $E^{root}$, we must have that every flow-path corresponding to a directed edge of $E^{root}$ consists entirely of dynamic links and has (graph-theoretic) length $\lceil \log_2(n)\rceil+3$.

Consider the flow-path corresponding to the directed edge $(r,x_i)\in E^{root}$. Suppose that the dynamic link $(l_k,v_j)$ appears on this flow-path (exactly one such dynamic link does). Note that there are at most $n$ such dynamic links used in the flow-paths of $F$ corresponding to the directed edges of $E^{root}$. Build the set of clauses $\mathcal{C}^\prime \subseteq \mathcal{C}$ by including every such $c_j$ in $\mathcal{C}^\prime$. Consequently, we have a subset $\mathcal{C}^\prime \subseteq \mathcal{C}$ of at most $n$ clauses. Moreover, as we can reach any $x_i$ from $r$ by a flow-path of $d+3$ dynamic links, every $x_i\in \mathcal{X}$ appears in a clause of $\mathcal{C}^\prime$; that is, $|\mathcal{C}^\prime| = n$, $\mathcal{C}^\prime$ is an exact cover for $\mathcal{X}$ and $(\mathcal{X},\mathcal{C})$ is a yes-instance of RXC3. As our instance $(H,\mu,D,\kappa)$ can be constructed from $(\mathcal{X},\mathcal{C})$ in time polynomial in $n$, our result follows.
\qed \end{proof}

These results led us to consider the problem \textsc{1-switched RRP($\sigma=k$)} for $k \ge 0$. By extending Lemma \ref{lem:FPS19-ptime} due to \cite{FPS19}, we show that this restriction leads to a tractable problem when either $\sigma=0$ or the static network is a complete graph, in contrast with our NP-completeness results. That result makes use of the \emph{dynamic link limit} $\delta$ - recall that when $\delta=1$ every flow-path must contain at most one dynamic link. In combination with the constraint $\sigma=0$, this entails that every flow path consists of either a single dynamic link (and no static links), or of one or more static links (and no dynamic links). The result below is proven via a maximum matching argument in \cite{FPS19}.

\begin{lemma}[\mbox{\cite[Theorem 3.1]{FPS19}}]\label{lem:FPS19-ptime}
	The problem RRP($\sigma=0$, $\delta=1$) can be solved in polynomial-time when there is only one switch. Moreover, the problem remains solvable in polynomial-time even if we allow non-uniform weights for dynamic links.
\end{lemma}

\begin{theorem}\label{thm:sigma0_poly}
    \textsc{$1$-switched RRP($\sigma=0$)} is in P. 
\end{theorem}

\begin{proof}
Note that this proof holds even if we allow variable weights for static and optic links. Because each node $v$ is connected to the switch exactly once, it is not possible that any vertex is incident to two optic links, and hence impossible for there to be any path consisting of two or more optic links. Consequently setting $\delta=1$ introduces no new constraints; in this scenario a flow is permissible with $\sigma=0$ if and only if it permissible with $\sigma=0$ and $\delta=1$. Hence the input instance of \textsc{1-switched RRP($\sigma=0, \delta=1$)} is a yes-instance if and only the corresponding instance of \textsc{RRP($\sigma=0, \delta=1$)} is a yes-instance, and tractability follows from Lemma \ref{lem:FPS19-ptime}.
\qed \end{proof}

\begin{corollary}\label{cor:cliqueP}
    \textsc{$1$-switched RRP($\sigma=k$)} restricted to instances where the static network $G$ is a complete graph is in P, for any $k \in \mathbb N \cup \{\infty\}$.
\end{corollary}

\begin{proof}

If $G$ is a complete graph with all edge weights equal (without loss take $wt(u,v)=1~\forall (u,v) \in E$) then without loss under any configuration $N$, each demand $D[u,v]$ is routed via the flow-path $\varphi(u,v)$ of minimum weight, which is either:
\begin{itemize}
    \item a single static link from $u$ to $v$ with unit weight. 
    \item a single dynamic link from $u$ to $v$ with weight $\mu$.
\end{itemize}

It follows that setting $\sigma=0$ introduces no new constraints, and then by Theorem \ref{thm:sigma0_poly} we have tractability.  \qed \end{proof}

Corollary \ref{cor:cliqueP} rules out the possibility that \textsc{1-switched RRP} might be NP-complete for any polynomial graph family $\mathcal H$ (since such a claim would extend to the family of complete graphs) unless P equals NP. This leaves open the practically relevant case where $\mathrm\Delta_S=1$ and $\sigma>0$ for \emph{specific} topologies. We consequently consider the scenario where the static network is a hypercube and the segregation parameter $\sigma=3$. 

\begin{theorem}\label{thm:hypercubeRRPsigma3}
	For any fixed $\mu \in (0,1)$, the problem \textsc{1-switched RRP($\sigma = 3$)} restricted to instances $(H, \mu, D, \kappa)$ satisfying:
	\begin{itemize}
		\item $H\in \mathcal Q$, where $\mathcal{Q}:= \{Q_d | d \in \mathbb N\}$ is the family of hypercubes
		\item the workload matrix $D$ is sparse and all values in it are polynomial in $n$
	\end{itemize}
	is NP-complete.
\end{theorem}

\begin{proof}\setcounter{claim}{0}
	Note that RRP as in the statement of the theorem is in NP on account of the numerical restrictions imposed. As in the proof of Theorem~\ref{thm:3switchedRRP}, we reduce from the problem RXC3. Let $(\mathcal{X},\mathcal{C})$ be an instance of RXC3 of size $n$. W.l.o.g., we may assume that $n$ is even. Let $m=\lceil \log_2(3n)\rceil$ (so, $3n \leq 2^m$). Our static network $H$ is the hypercube $Q_{8m}$ (so, $|Q_{8m}| = 2^{8m} = O(n^8)$) and there is exactly $1$ switch link from every node to our switch $s$.
	
	We begin by defining a weighted digraph $D^\prime = (V,E^\prime)$ where $V$ is the node set of our hypercube $Q_{8m}$. The digraph $D^\prime$ is to provide a description of our workload matrix. In order to define the directed edges of $D^\prime$, we define some specific subsets of nodes of $V$. We explain the notation that we use as we proceed. By `distance' we mean the distance (that is, shortest path length) in the hypercube $Q_{8m}$ and by $V_i$ (resp. $V_{\leq i}$) we mean the subset of nodes of $V$ that are at distance exactly $i$ (resp. at most $i$) from the \emph{root node} $r=(1, 1, \ldots, 1) = 1^{8m}$ (in general, we also denote nodes of $V$ as bit-strings of length $8m$).
	\begin{itemize}
		\item Within the nodes of $V_m$, we define the subset $\tilde{P}=\{z\bar{z}1^{6m}: z\in\{0,1\}^m\}$ (in general: $yz$ denotes the concatenation of two bit-strings $y$ and $z$, with $z^i$ denoting the concatenation of $i$ copies of the bit-string $z$; and $\bar{z}$ denotes the complement of the bit-string $z$). Note that $|\tilde{P}|=2^m$.
		\item Within the nodes of $V_{3m}$, we define the subset $\tilde{C} = \{(z\bar{z})^31^{2m}: z\in\{0,1\}^m\}$. Note that $|\tilde{C}|=2^m$.
		\begin{itemize}
			\item For any bit-string $z\in\{0,1\}^{8m}$ and for any bit-string $y\in\{0,1\}^i$ where $i < 8m$, we write $\oplus(z;y)$ to denote the bitwise exclusive OR of $z$ and $0^{8m-i}y$ (so, $\oplus(z;y)$ differs from $z$ on at most the last $i$ bits). For any set $U$ of bit-strings of length $8m$ and bit-string $y$ of length $i < 8m$, we define $U^{\oplus y} = \{\oplus(z;y): z\in U\}$.
			\item For any $c\in \tilde{C}$, we define
			\begin{itemize}
				\item $c^{\oplus 001} = \oplus(c; 001)$
				\item $c^{\oplus 010} = \oplus(c; 010)$
				\item $c^{\oplus 100} = \oplus(c; 100)$
			\end{itemize}
			and so we obtain sets of nodes $\tilde{C}^{\oplus 001}$, $\tilde{C}^{\oplus 010}$ and $\tilde{C}^{\oplus 100}$, each consisting of sets of neighbours of nodes in $\tilde{C}$ (note also that if $c_1,c_2\in \tilde{C}$ are distinct then $c_1$ and $c_2$ are at distance at least $6$ in $Q_{8m}$ and any node of $\{c_1^{\oplus 001}, c_1^{\oplus 010}, c_1^{\oplus 100}\}$ and node of $\{c_2^{\oplus 001}, c_2^{\oplus 010}, c_2^{\oplus 100}\}$ are at distance at least $6$ in $Q_{8m}$).
		\end{itemize} 
		\item Within the nodes of $V_{4m}$, we define the subset $\tilde{X} = \{(z\bar{z})^4: z\in\{0,1\}^m\}$. Note that $|\tilde{X}|=2^m$.
		\begin{itemize}
			\item For any $x\in \tilde{X}$, we define
			\begin{itemize}
				\item $x^{\oplus 01} = \oplus(x; 01)$
				\item $x^{\oplus 10} = \oplus(x; 10)$
			\end{itemize}
		and so we obtain sets of nodes $\tilde{X}^{\oplus 01}$ and $\tilde{X}^{\oplus 10}$, each consisting of sets of neighbours of nodes in $\tilde{X}$ (note that if $x_1,x_2\in \tilde{X}$ are distinct then $x_1$ and $x_2$ are at distance at least $8$ in $Q_{8m}$ and any node of $\{x_1^{\oplus 01}, x_1^{\oplus 10}\}$ and node of $\{x_2^{\oplus 01}, x_2^{\oplus 10}\}$ are at distance at least $6$ in $Q_{8m}$).
		\end{itemize}
	\end{itemize}

	We define the weighted directed edges $E^\prime$ of $D^\prime$ as follows. There are three weights: some $\beta > \max{\{\frac{3n(m+1+2\mu)}{\mu},\frac{3n(m+1+2\mu)}{1-\mu}\}}$; some $\alpha > 15n\beta + 9n^2 + 9n$; and 1. The directed edges of $E_\beta$ (resp. $E_\alpha$, $E_1$) all have weight $\beta$ (resp. $\alpha$, $1$) and $E^\prime = E_\beta\cup E_\alpha \cup E_1$.
	\begin{itemize}
	\item The directed edges of $E_\beta$ are derived using the structure of our instance $(\mathcal{X},\mathcal{C})$ of RXC3.
	\begin{itemize}
		\item For each $1\leq j\leq 3n$, identify the clause $c_j\in\mathcal{C}$ with the node $c_j=(bin_m(j - 1)\widebar{bin_m(j - 1)})^31^{2m}$ of the subset $\tilde{C}$ of $V_{3m}$, where in general $bin_i(j)$ is the binary representation of the natural number $j$ as a bit-string of length $i$ (so, in particular, $0\leq j< 2^i$). Henceforth, we refer to these specific nodes of $\tilde{C}$ as the set of \emph{clause nodes} $C = \{c_j: 1\leq j \leq 3n\}$.
		\item For each $1\leq i\leq 3n$, identify the element $x_i$ of $\mathcal{X}$ with the node $x_i = (bin_m(i - 1)\widebar{bin_m(i - 1)})^4$ of the subset $\tilde{X}$ of $V_{4m}$. Henceforth, we refer to these specific nodes of $\tilde{X}$ as the set of \emph{element nodes} $X = \{x_i:1\leq i\leq 3n\}$.
	\end{itemize}
	For any clause $c_j=\{x_1^j,x_2^j,x_3^j\}$ of $\mathcal{C}$, there are directed edges $(c_j^{\oplus 001}, x_1^i)$, $(c_j^{\oplus 010}, x_2^i)$ and $(c_j^{\oplus 100}, x_3^i)$ in $E_\beta$ and we refer to the nodes $c_j^{\oplus 001}$, $c_j^{\oplus 010}$ and $c_j^{\oplus 100}$ as the \emph{associate clause nodes} of clause $c_j$ or of clause node $c_j$. We define the \emph{associate element nodes} of an element or an element node analogously. This constitutes all directed edges of $E_\beta$.
	
	\item We denote by $P=\{p_i: 1\leq i\leq n\}$ a subset of $n$ nodes of $\tilde{P}$ which we refer to as the set of \emph{port nodes}. The directed edges of $E_\alpha$ are $E_\alpha^{P} \cup E_\alpha^{W}$ where:
	\begin{itemize}
		\item $E_\alpha^{P} = \{(z,\bar{z}) : z\in V_{\leq m+5}\setminus {P}\}$
		\item $E_\alpha^{W}$ is an arbitrary orientation of a specific perfect matching $M$ (coming up) on the nodes of  $W\setminus (X\cup X^{\oplus 01} \cup X^{\oplus 10})$ where $W$ is defined as  $V_{4m} \cup\bigcup_{i=1}^{3}V_{4m-i} \cup \bigcup_{i=1}^{3}V_{4m+i}$.
	\end{itemize}
	Note that $|W\setminus (X\cup X^{\oplus 01} \cup X^{\oplus 10})|$ is even and so a perfect matching exists (recall, we assumed that $n$ is even); however, as stated, we require that our perfect matching $M$ has a specific property which we describe below.
	
	\item The set of directed edges $E_1$ is defined as $\{(r, x_i): 1\leq i\leq 3n\}$; that is, as $\{r\}\times X$.
\end{itemize}

Now for our perfect matching $M$.
	
\begin{claim}\label{clm:matching}
	There is a perfect matching $M$ of $W\setminus (X\cup X^{\oplus 01} \cup X^{\oplus 10})$ so that if $(u,v)$ is in the matching then there is no static link joining $u$ and $v$.
\end{claim}

\begin{proof}
	Define a perfect matching $M$ on $W$ as follows. First, match every node in $V_{4m}$ with the node $v$ that differs from $u$ in every bit; that is, $v=\bar{u}$.
	
	Consider some node $z\bar{z}z\bar{z}z\bar{z}z\bar{z}$ of $X$. Its matched node is $\bar{z}z\bar{z}z\bar{z}z\bar{z}z$ which is either in $X$ or outside $X$. Suppose it is the latter. Suppose also that some other node $w\bar{w}w\bar{w}w\bar{w}w\bar{w}$ is in $X$ and its matched node $\bar{w}w\bar{w}w\bar{w}w\bar{w}w$ lies outside $X$. Note that the distance between the two matched nodes $\bar{z}z\bar{z}z\bar{z}z\bar{z}z$ and $\bar{w}w\bar{w}w\bar{w}w\bar{w}w$ is at least $8$. We amend our matching $M$ by removing all matched pairs of nodes of $X$ and choose an arbitrary matching on the remaining matched nodes (note that there is an even number of such matched nodes). Any pair of nodes in $M$ is such that there is no static link joining them and every node of $V_{4m}\setminus X$ is involved in $M$.

	Extend $M$ with a matching so that every node of $V_{4m-1}\setminus (X^{\oplus 01}\cup X^{\oplus 10})$ is matched with a node in $V_{4m+1}\setminus (X^{\oplus 01}\cup X^{\oplus 10})$. This is possible as $|V_{4m-1}|=|V_{4m+1}|$ and $|V_{4m-1}\cap (X^{\oplus 01}\cup X^{\oplus 10})| = |V_{4m+1}\cap (X^{\oplus 01}\cup X^{\oplus 10})|$. Further extend $M$ with a matching so that: every node of $V_{4m-2}$ (resp. $V_{4m-3}$) is matched with a node in $V_{4m+2}$ (resp. $V_{4m+3}$). The claim follows.
\qed \end{proof}

The digraph $D^\prime$ can be visualized in Fig.~\ref{fig:Dprime2}. The rows are intended to illustrate nodes in different $V_i$'s, although the rows labelled $V_{3m-1/3m+1}$ and $V_{4m-1/4m+1}$ contain all nodes of $V_{3m-1}\cup V_{3m+1}$ and $V_{4m-1}\cup V_{4m+1}$, respectively (this is where the associate clause nodes and the associate element nodes lie, respectively). We have not shown all directed edges; just enough to give a flavour of the construction. Directed edges from $E_\beta$ are shown as solid directed edges with the directed edges corresponding to the clause $c_1=\{x_1,x_n,x_2\}$, for example, depicted (the grey nodes are the associate clause nodes in order $c_1^{\oplus 001}, c_1^{\oplus 010}, c_1^{\oplus 100},c_2^{\oplus 001},\ldots$ from left to right). Directed edges from $E_\alpha$ are shown as dotted directed edges and the grey rectangle contains all nodes from $W$ involved in directed edges of $E_\alpha$. The white nodes within this rectangle are the element nodes and the associate element nodes (that is, the nodes of $W$ that are not incident with directed edges from $E_\alpha$). Finally, the dashed directed edges depict the directed edges from $E_1$ and the port nodes and clause nodes are also shown in white.

	\begin{figure*}
		\centering
		\scalebox{1}{
			\begin{tikzpicture}
				\node[state,circle,scale=0.3,fill=black] (r) at (0,0) {};
				\node at (0,0.4) {$root = r$};
				\node[state,circle,scale=0.3,fill=black] (c11) at (-1.25,-0.5) {};
				\node[state,circle,scale=0.3,fill=black] (c12) at (1.25,-0.5) {};
				\node at (5.5,-0.5) {$V_1$};
				\node[state,circle,scale=0.3,fill=black] (c21) at (-1.75,-1) {};
				\node[state,circle,scale=0.3,fill=black] (c22) at (-0.75,-1) {};
				\node[state,circle,scale=0.3,fill=black] (c23) at (0.75,-1) {};
				\node[state,circle,scale=0.3,fill=black] (c24) at (1.75,-1) {};
				\node at (5.5,-1) {$V_2$};
				\node at (-1.75,-1.5) {$\ldots$};
				\node at (1.75,-1.5) {$\ldots$};
				\node[state,circle,scale=0.3,fill=white] (c31) at (-3.5,-2) {};
				\node at (-3.5,-1.7) {$p_1$};
				\node[state,circle,scale=0.3,fill=white] (c32) at (-2.5,-2) {};
				\node at (-2.5,-1.7) {$p_2$};
				\node at (-1.75,-2) {$\ldots$};
				\node[state,circle,scale=0.3,fill=white] (c33) at (-1.0,-2) {};
				\node at (-1.15,-1.7) {$p_n$};
				\node[state,circle,scale=0.3,fill=black] (c34) at (0,-2) {};
				\node at (0.15,-1.7) {$z$};
				\node at (0.75,-2) {$\ldots$};
				\node[state,circle,scale=0.3,fill=black] (c35) at (1.5,-2) {};
				\node at (1.5,-1.7) {$z$};
				\node[state,circle,scale=0.3,fill=black] (c36) at (2.5,-2) {};
				\node at (2.5,-1.7) {$z$};
				\node at (3.25,-2) {$\ldots$};
				\node[state,circle,scale=0.3,fill=black] (c37) at (4.00,-2) {};
				\node at (4,-1.7) {$z$};
				\node at (5.5,-2) {$V_m$};
				\node[state,circle,scale=0.3,fill=white] (c41) at (-3.5,-4) {};
				\node at (-3.7,-3.7) {$c_1$};
				\node[state,circle,scale=0.3,fill=white] (c42) at (-2.5,-4) {};
				\node at (-2.7,-3.7) {$c_2$};
				\node at (-1.75,-4) {$\ldots$};
				\node[state,circle,scale=0.3,fill=white] (c43) at (-1.0,-4) {};
				\node at (-1.05,-3.7) {$c_n$};
				\node[state,circle,scale=0.3,fill=white] (c44) at (-0.0,-4) {};
				\node at (.45,-3.7) {$c_{n+1}$};
				\node at (0.75,-4) {$\ldots$};
				\node[state,circle,scale=0.3,fill=white] (c45) at (1.5,-4) {};
				\node at (1.85,-3.7) {$c_{3n}$};
				\node[state,circle,scale=0.3,fill=black] (c46) at (2.5,-4) {};
				\node at (3.25,-4) {$\ldots$};
				\node[state,circle,scale=0.3,fill=black] (c47) at (4.00,-4) {};
				\node at (5.5,-4) {$V_{3m}$};
				\node[state,circle,scale=0.3,fill=lightgray] (c51) at (-3.85,-4.5) {};
				\node[state,circle,scale=0.3,fill=lightgray] (c52) at (-3.5,-4.5) {};
				\node[state,circle,scale=0.3,fill=lightgray] (c53) at (-3.15,-4.5) {};
				\node[state,circle,scale=0.3,fill=lightgray] (c54) at (-2.85,-4.5) {};
				\node[state,circle,scale=0.3,fill=lightgray] (c55) at (-2.5,-4.5) {};
				\node[state,circle,scale=0.3,fill=lightgray] (c56) at (-2.15,-4.5) {};
				\node[state,circle,scale=0.3,fill=lightgray] (c57) at (-1.35,-4.5) {};
				\node[state,circle,scale=0.3,fill=lightgray] (c58) at (-1.0,-4.5) {};
				\node[state,circle,scale=0.3,fill=lightgray] (c59) at (-0.65,-4.5) {};
				\node[state,circle,scale=0.3,fill=lightgray] (c5a) at (-0.35,-4.5) {};
				\node[state,circle,scale=0.3,fill=lightgray] (c5b) at (0,-4.5) {};
				\node[state,circle,scale=0.3,fill=lightgray] (c5c) at (0.35,-4.5) {};
				\node[state,circle,scale=0.3,fill=lightgray] (c5d) at (1.15,-4.5) {};
				\node[state,circle,scale=0.3,fill=lightgray] (c5e) at (1.5,-4.5) {};
				\node[state,circle,scale=0.3,fill=lightgray] (c5f) at (1.85,-4.5) {};
				\node[state,circle,scale=0.3,fill=black] (c5g) at (2.15,-4.5) {};
				\node[state,circle,scale=0.3,fill=black] (c5h) at (2.5,-4.5) {};
				\node[state,circle,scale=0.3,fill=black] (c5i) at (2.85,-4.5) {};
				\node[state,circle,scale=0.3,fill=black] (c5j) at (3.65,-4.5) {};
				\node[state,circle,scale=0.3,fill=black] (c5k) at (4.0,-4.5) {};
				\node[state,circle,scale=0.3,fill=black] (c5l) at (4.35,-4.5) {};
				\node at (5.5,-4.5) {$V_{3m-1/3m+1}$};
				\node[state,rectangle,scale=1,fill=lightgray, minimum width = 8.6cm, minimum height = 1.7cm] (rect1) at (0.2,-6.75) {};
				\node[state,circle,scale=0.3,fill=white] (c61) at (-3.75,-6.5) {};
				\node[state,circle,scale=0.3,fill=white] (c62) at (-3.25,-6.5) {};
				\node[state,circle,scale=0.3,fill=white] (c63) at (-2.75,-6.5) {};
				\node[state,circle,scale=0.3,fill=white] (c64) at (-2.25,-6.5) {};
				\node[state,circle,scale=0.3,fill=white] (c65) at (-1.25,-6.5) {};
				\node[state,circle,scale=0.3,fill=white] (c66) at (-0.75,-6.5) {};
				\node[state,circle,scale=0.3,fill=white] (c67) at (-0.25,-6.5) {};
				\node[state,circle,scale=0.3,fill=white] (c68) at (0.25,-6.5) {};
				\node[state,circle,scale=0.3,fill=white] (c69) at (1.25,-6.5) {};
				\node[state,circle,scale=0.3,fill=white] (c6a) at (1.75,-6.5) {};
				\node[state,circle,scale=0.3,fill=black] (c6b) at (2.25,-6.5) {};
				\node[state,circle,scale=0.3,fill=black] (c6c) at (2.75,-6.5) {};
				\node[state,circle,scale=0.3,fill=black] (c6d) at (3.75,-6.5) {};
				\node[state,circle,scale=0.3,fill=black] (c6e) at (4.25,-6.5) {};
				\node at (5.5,-6.5) {$V_{4m-1/4m+1}$};
				\node[state,circle,scale=0.3,fill=white] (c71) at (-3.5,-7) {};
				\node at (-3.5,-7.4) {$x_1$};
				\node[state,circle,scale=0.3,fill=white] (c72) at (-2.5,-7) {};
				\node at (-2.5,-7.4) {$x_2$};
				\node at (-1.75,-7) {$\ldots$};
				\node[state,circle,scale=0.3,fill=white] (c73) at (-1.0,-7) {};
				\node at (-1.0,-7.4) {$x_n$};
				\node[state,circle,scale=0.3,fill=white] (c74) at (0.0,-7) {};
				\node at (-0,-7.4) {$x_{n+1}$};
				\node at (0.75,-7) {$\ldots$};
				\node[state,circle,scale=0.3,fill=white] (c75) at (1.5,-7) {};
				\node at (1.5,-7.4) {$x_{3n}$};
				\node[state,circle,scale=0.3,fill=black] (c76) at (2.5,-7) {};
				\node at (3.25,-7) {$\ldots$};
				\node[state,circle,scale=0.3,fill=black] (c77) at (4.0,-7) {};
				\node at (5.5,-7) {$V_{4m}$};
				\node[state,circle,scale=0.3,fill=black] (c81) at (2.5,-9) {};
				\node at (2.5,-9.4) {$\bar{z}$};
				\node[state,circle,scale=0.3,fill=black] (c82) at (4.0,-9) {};
				\node at (4,-9.4) {$\bar{z}$};
				\node[state,circle,scale=0.3,fill=black] (c83) at (1.5,-9) {};
				\node at (1.5,-9.4) {$\bar{z}$};
				\node[state,circle,scale=0.3,fill=black] (c84) at (0,-9) {};
				\node at (0,-9.4) {$\bar{z}$};
				\node at (5.5,-9) {$V_{7m}$};
				\node[state,circle,scale=0.3,fill=black] (c91) at (-1.25,-10.5) {};
				\node at (-1.25,-10.9) {$\bar{z}$};
				\node[state,circle,scale=0.3,fill=black] (c92) at (1.25,-10.5) {};
				\node at (1.25,-10.9) {$\bar{z}$};
				\node at (5.5,-0.5) {$V_1$};
				\node[state,circle,scale=0.3,fill=black] (c8x1) at (-1.75,-10) {};
				\node at (-1.75,-10.4) {$\bar{z}$};
				\node[state,circle,scale=0.3,fill=black] (c8x2) at (-0.75,-10) {};
				\node at (-0.75,-10.4) {$\bar{z}$};
				\node[state,circle,scale=0.3,fill=black] (c8x3) at (0.75,-10) {};
				\node at (0.75,-10.4) {$\bar{z}$};
				\node[state,circle,scale=0.3,fill=black] (c8x4) at (1.75,-10) {};
				\node at (1.75,-10.4) {$\bar{z}$};
				\node[state,circle,scale=0.3,fill=black] (compr) at (-0.0,-11) {};
				\node at (0,-11.4) {$\bar{r}$};
				\path (r) edge [dashed,->] (c71);
				\path (r) edge [dashed,->] (c72);
				\path (r) edge [dashed,->] (c73);
				\path (r) edge [dashed,->] (c74);
				\path (r) edge [dashed,->] (c75);
				\draw[densely dotted, ->] (c36) to[out=290,in=70] (c81);
				\draw[densely dotted, ->] (c37) to[out=290,in=70] (c82);
				\draw[densely dotted, ->] (c35) to[out=290,in=70] (c83);
				\draw[densely dotted, ->] (c34) to[out=290,in=70] (c84);
				\path (c11) edge[densely dotted,->] (-1.25,-1.4);
				\path (c12) edge[densely dotted,->] (1.25,-1.4);
				\path (c21) edge[densely dotted,->] (-1.75,-1.4);
				\path (c22) edge[densely dotted,->] (-0.75,-1.4);
				\path (c23) edge[densely dotted,->] (0.75,-1.4);
				\path (c24) edge[densely dotted,->] (1.75,-1.4);
				\path (-1.25,-8.1) edge[densely dotted,->] (c91);
				\path (1.25,-8.1) edge[densely dotted,->] (c92);
				\path (-1.75,-8.1) edge[densely dotted,->] (c8x1);
				\path (-0.75,-8.1) edge[densely dotted,->] (c8x2);
				\path (0.75,-8.1) edge[densely dotted,->] (c8x3);
				\path (1.75,-8.1) edge[densely dotted,->] (c8x4);
				\draw[densely dotted, ->] (r) to[out=330,in=90] (0.5,-1.4);
				\draw[densely dotted, ->] (0.5,-8.1) to[out=270,in=60] (compr);
				\path (c51) edge[->] (c71);
				\path (c52) edge[->] (c73);
				\path (c53) edge[->] (c72);
			\end{tikzpicture}
		}
		\caption{The graph $D^\prime$.}
		\label{fig:Dprime2}
	\end{figure*}
	
	In order to complete our instance $(H, \mu, D,\kappa)$, we define $\mu$ to be any fixed (rational) value in the interval $(0,1)$ and $\kappa = \kappa_\beta + \kappa_\alpha + \kappa_1$ where:
	\begin{itemize}
		\item $\kappa_\beta = \frac{|E_\beta|}{3}(\mu\beta) + \frac{2|E_\beta|}{3}((\mu + 1)\beta) = 3n\mu\beta + 6n(\mu+1)\beta$
		\item $\kappa_\alpha = |E_\alpha|\mu\alpha$
		\item $\kappa_1 = |E_1|(m+1+2\mu) = 3n(m+1+2\mu)$.
	\end{itemize}
	
	Suppose that there is an exact cover $\mathcal{C}^\prime=\{c_{j_i}:1\leq i\leq n\}\subseteq \mathcal{C}$ of $\mathcal{X}$. Let the set of clause nodes $C^\prime \subseteq C$ be $\{c_{j_i}: 1\leq i\leq n\}$. We define a configuration $N$ by choosing dynamic links as follows.
	\begin{itemize}
		\item For some arbitrary bijection $f$ from $P$ to $C^\prime$, $\{(p,f(p)): p \in P\}$ is a set of $n$ dynamic links.
		\item For each clause $c_j = \{x_1^j,x_2^j,x_3^j\}\in \mathcal{C}^\prime$, there are dynamic links $(c_j^{\oplus 001}, x_1^j)$, $(c_j^{\oplus 010}, x_2^j)$ and $(c_j^{\oplus 100}, x_3^j)$; so, all nodes of $X$ are incident with a dynamic link as are all associate clause nodes of $\{c_j^{\oplus 001}, c_j^{\oplus 010},c_j^{\oplus 100} : c_j\in C^\prime\}$. Note that these dynamic links result from the directed edges of $E_\beta$ incident with the associate clause nodes of the clause nodes of $C^\prime$.
		\item All directed edges of $E_\alpha$ result in dynamic links. In particular, there is no dynamic link from a node outside $W$ to a node inside $W$ except possibly incident with the element nodes and the associate element nodes; indeed, the element nodes are already incident with such `external' dynamic links.
		\item Consider the node $x_1\in X$. The element $x_1$ appears in two clauses of $\mathcal{C}\setminus \mathcal{C}^\prime$, say $c_{i_1}$ and $c_{i_2}$. Suppose, for example, that the element $x_1$ appears as the second element of clause $c_{i_1}$ and as the first element of clause $c_{i_2}$. If so then include a dynamic link from each of $c_{i_1}^{\oplus 010}$ and $c_{i_2}^{\oplus 001}$ to $x_1^{\oplus 01}$ and $x_1^{\oplus 10}$. Alternatively, if, say, the element $x_1$ appears as the third element of clause $c_{j_1}$ and as the third element of clause $c_{j_2}$ then include a dynamic link from each of $c_{j_1}^{\oplus 100}$ and $c_{j_2}^{\oplus 100}$ to $x_1^{\oplus 01}$ and $x_1^{\oplus 10}$. Proceed similarly and analogously with all remaining element nodes of $X\setminus\{x_1\}$. On completion of this iterative process, there is a dynamic link from every associate clause node to a unique element node or associate element node; that is, these dynamic links depict a bijection from the associate clause nodes to the element nodes and the associate element nodes.
	\end{itemize}
	This constitutes the configuration $N$.
	
	Consider some $\beta$-demand of our workload. For some node $c_j\in C^\prime$, all such workloads originating at the nodes $c_j^{\oplus 001}$, $c_j^{\oplus 010}$ and $c_j^{\oplus 100}$ can be served via a flow-path consisting of a solitary dynamic link at workload cost $\mu\beta$. For some node $c_j\in C\setminus C^\prime$, suppose that there is a demand $D[c_j^{\oplus 001}, x_1^i] = \beta$. This demand exists because the element $x_i$ is the first element of clause $c_j$. As the element $x_i$ is the first element of clause $c_j$, there is a dynamic link from $c_j^{\oplus 001}$ to a neighbour of $x_i$ in $\{x_i^{\oplus 01},x_i^{\oplus 10}\}$, which w.l.o.g. we may assume to be $x_i^{\oplus 01}$; hence, the demand can be served via the flow-path $c_j^{\oplus 001}, x_i^{\oplus 01}, x_i$ at workload cost $(\mu + 1)\beta$. Defining other flow-paths analogously means that we can find flow-paths corresponding to the directed edges of $E_\beta$ the total workload cost of which is $3n\mu\beta + 6n(\mu + 1)\beta = \kappa_\beta$.
	
	All $\alpha$-demands can be served via a flow-path consisting of a solitary dynamic link at total workload cost $|E_\alpha|\mu\alpha = \kappa_\alpha$. 
	
	Consider some $1$-demand $D[r,x_i]=1$. Let the clause $c_j\in \mathcal{C}^\prime$ be such that $x_i$ is an element of $c_j$; hence, there is a dynamic link from some port node $p_{j^\prime}$ of $P$ to $c_j$ and a dynamic link from some neighbour of $c_j$ from $\{c_j^{\oplus 001}, c^{\oplus 010}, c^{\oplus 110}\}$ to $x_i$. Consequently, there is a static path from $r$ to $p_{j^\prime}$ of weight $m$, a dynamic link from $p_{j^\prime}$ to the node $c_j$, a static link from $c_j$ to the above neighbour in $\{c_j^{\oplus 001}, c^{\oplus 010}, c^{\oplus 110}\}$ and a dynamic link from this neighbour to $x_i$. Hence, the demand $D[r,x_i]=1$ can be served via a flow-path of workload cost $m+1+2\mu$ with all $1$-demands served via flow-paths at a total workload cost $3n(m+1+2\mu) = \kappa_1$. Consequently, our instance $(H, \mu,D,\kappa)$ is a yes-instance of \textsc{1-switched RRP($\sigma=3$)}.
	
	Conversely, suppose it is the case that $(H, \mu,D,\kappa)$ is a yes-instance of \linebreak\textsc{1-switched RRP($\sigma=3$)} and that $N$ is a configuration and $F$ a set of flow-paths witnessing that the total workload cost is at most $\kappa$. 
	
	\begin{claim}\label{allEalpha}
		Every directed edge $(u,v)$ of $E_\alpha$ is necessarily such that $(u,v)\in N$ and the total workload cost of flow-paths serving the $\alpha$-demands is exactly $\kappa_\alpha$.
	\end{claim}

\begin{proof}
	Suppose that at least one of the $\alpha$-demands is served via a flow-path at a workload cost of more than $\mu\alpha$; so, it must be at a workload cost of more than $(1+\mu)\alpha$ as we cannot traverse a dynamic link followed immediately by another dynamic link (recall, $\Delta_S=1$) and by Claim~\ref{clm:matching}, if $D[u,v]=\alpha$ then there is no static link $(u,v)$. Hence, the total workload cost of flow-paths serving the $\alpha$-demands is at least $(|E_\alpha| - 1)\mu\alpha + (1+\mu)\alpha = \kappa_\alpha +\alpha$. We have that $\kappa_\beta + \kappa_1 = 3n\mu\beta + 6n(\mu + 1)\beta + 3n(m+1+2\mu) < 3n\beta + 12n\beta + 3n(3n+3) = 15n\beta + 9n^2+9n < \alpha$. Hence, the total workload cost of all flow-paths of $F$ is strictly greater than $\kappa$ which yields a contradiction and the claim follows.
\qed \end{proof}
Call the dynamic links corresponding to the directed edges of $E_\alpha$ the \emph{$\alpha$-dynamic links}.

\begin{claim}\label{allEbeta}
	Exactly $3n$ (resp. $6n$) $\beta$-demands are served by a flow-path of workload cost $\mu\beta$ (resp. $(1+\mu)\beta$) exactly; that is, exactly $3n$ (resp. $6n$) flow-paths serving the $\beta$-demands consist of a dynamic link (resp. a dynamic link and a static link). Hence, the total workload cost of the flow-paths serving $\beta$-demands is exactly $\kappa_\beta$.
\end{claim}

\begin{proof}
All $\beta$-demands are from a unique associate clause node to an element node. As $\Delta_S=1$, at most $3n$ dynamic links can be incident with an element node. So, at most $3n$ $\beta$-demands are served by a flow-path consisting of a solitary dynamic link. If a $\beta$-demand is served by a flow-path that does not consist of a solitary dynamic link then the flow-path has workload cost at least $(1+\mu)\beta$. 

Suppose there are less than $3n$ $\beta$-demands that are served by a flow-path consisting of a solitary dynamic link. So, the total workload cost of flow-paths serving $\beta$-demands is at least $(3n-1)\mu\beta + (6n+1)(1+\mu)\beta = 9n\mu\beta + 6n\beta + \beta$. By Claim~\ref{allEalpha}, $9n\mu\beta + 6n\beta + \beta \leq \kappa_\beta+\kappa_1 = 3n\mu\beta + 6n(\mu+1)\beta + 3n(m+1+2\mu)$; that is, $\beta \leq 3n(m+1+2\mu)$ which yields a contradiction. So, there are exactly $3n$ $\beta$-demands served by a flow-path consisting of a solitary dynamic link. 

Suppose that there is a $\beta$-demand served by a flow-path of workload cost neither $\mu\beta$ nor $(1+\mu)\beta$. Such a flow-path has workload cost at least $\gamma\beta$ where $\gamma = \min\{2, 1+2\mu\}$. As before, the total workload cost of flow-paths serving $\beta$-demands is at least $3n\mu\beta + (6n-1)(1+\mu)\beta + \gamma\beta = 9n\mu\beta + 6n\beta + (\gamma - 1 - \mu)\beta$. By Claim~\ref{allEalpha}, $9n\mu\beta + 6n\beta + (\gamma - 1 - \mu)\beta \leq \kappa_\beta+\kappa_1 = 3n\mu\beta + 6n(\mu+1)\beta + 3n(m+1+2\mu)$; that is, $(\gamma - 1 - \mu)\beta \leq 3n(m+1+2\mu)$. If $\gamma = 2$ then $\beta \leq \frac{3n(m+1+2\mu)}{1-\mu}$; and if $\gamma= 1+2\mu$ then $\beta \leq \frac{3n(m+1+2\mu)}{\mu}$. Whichever is the case, we obtain a contradiction. Hence, there are exactly $6n$ $\beta$-demands served by a flow-path of workload cost $(1+\mu)\beta$. Each of these flow-paths consists of a dynamic link and a static link and the claim follows.
\qed \end{proof}

By Claim~\ref{allEbeta}, each element node $x$ is incident via a dynamic link to a unique associate clause node of some clause $c$ so that the element $x\in\mathcal{X}$ is in the clause $c\in\mathcal{C}$. Consider some $\beta$-demand served by a flow-path of weight $1+\mu$; that is, a flow-path $f$ consisting of a static link and a dynamic link. By Claim~\ref{allEalpha}, every node of $W\setminus(X\cup X^{\oplus 01} \cup X^{\oplus 10})$ is incident with a dynamic link incident with some other node of $W\setminus(X\cup X^{\oplus 01} \cup X^{\oplus 10})$. So, no node of $W\setminus(X\cup X^{\oplus 01} \cup X^{\oplus 10})$ can appear on $f$. Also, no associate clause node is adjacent via a static link to any other associate clause node. Hence, $f$ must be of the form $c^\prime, x^\prime, x_i$, where $c^\prime$ is an associate clause node, of some clause $c$, that is adjacent via a dynamic link to the associate element node $x^\prime$ which is adjacent via a static link to the element node $x$ and where the element $x\in\mathcal{X}$ is in the clause $c\in\mathcal{C}$. In particular:
\begin{itemize}
	\item for every $x_i\in \mathcal{X}$, if $x_i\in c_1^i,c_2^i,c_3^i$, with $c_1^i,c_2^i,c_3^i \in \mathcal{C}$, then there are $3$ dynamic links incident with a unique node of $\{x_i,x_i^{\oplus 01},x_i^{\oplus 10}\}$ and incident with exactly one node of each of the sets of associate clause nodes of $c_1^i$, $c_2^i$ and $c_3^i$
	\item if there is a dynamic link from an element node $x$ or one of its associate element nodes to an associate clauses node of some clause $c$ then the element $x\in\mathcal{X}$ is in the clause $c\in\mathcal{C}$.
\end{itemize}
Call these dynamic links the \emph{$\beta$-dynamic links}. Also, by Claim~\ref{allEbeta}, the total workload cost of the flow-paths serving the $\beta$-demands is $\kappa_\beta$. Our aim now is to show that there are $n$ clauses with the property that every associate clause node of any of these clauses is incident with a $\beta$-dynamic link incident with an element node. Doing so would result in our instance of RXC3 being a yes-instance.

Consider some flow-path $f$ that services some $1$-demand. Our preliminary aim is to show that every such flow-path $f$ has weight at least $m+1+2\mu$.

Suppose that $f$ involves $2$ dynamic links. Hence, the structure of $f$ is $s^\ast d s^+ d$ or $d s^+ d s^\ast$, where $s$ (resp. $d$) denotes a static (resp. dynamic) link and $\ast$ (resp. $+$) denotes at least $0$ (resp. at least $1$) occurrence (in a regular-language style); recall that $\sigma=3$.
\begin{itemize}
	\item If $f$ has structure $s^\ast d s^+ d$ then the second dynamic link must be of the form $(c^\prime,x)$, where $c^\prime$ is an associate clause node and $x$ is an element node. If the first dynamic link is not an $\alpha$- or $\beta$-dynamic link then the initial prefix of static links has weight at least $m$. Hence, the total weight of $f$ is at least $m+1+2\mu$. Alternatively, suppose that the first dynamic link is an $\alpha$- or $\beta$-dynamic link.
	\begin{itemize}
		\item If it is an $\alpha$-dynamic link and is of the form $(z,\bar{z})$, for some $z \in V_{\leq m}\setminus P$, then the total weight of $f$ is at least $i + \mu + 5m-i-1 +\mu$, for some $0\leq i \leq m$; that is, at least $5m-1+2\mu > m+1+2\mu$. 
		\item If it is an $\alpha$-dynamic link and is of the form $(u,v)$, for some $u,w\in W$, then the total weight of $f$ is at least $4m-1+\mu +1+\mu= 4m+2\mu > m+1+2\mu$.
		\item If it is a $\beta$-dynamic link then the total weight of $f$ is at least $3m-1+\mu + 1 + \mu =3m+2\mu> m+1+2\mu$.
	\end{itemize}
	\item If $f$ has structure $d s^+ d s^\ast$ then the first dynamic link must be $(r,\bar{r})$. 
	\begin{itemize}
		\item Suppose that the first sub-path of static links has weight less than or equal to $m$ with the second dynamic link being an $\alpha$-dynamic link of the form $(\bar{z},z)$. Then the total weight of $f$ is at least $\mu + i + \mu + 4m-i$, for some $0\leq i< m$; that is, $4m+2\mu > m+1+2\mu$.
		\item Suppose that the first sub-path of static links has weight exactly $m$ with the second dynamic link not an $\alpha$- or $\beta$-dynamic link. As all nodes of $W$ are already incident with a dynamic link, we must have that the weight of $f$ is at least $\mu+m+\mu+4 = m+4+2\mu > m+1+2\mu$.
		\item Suppose that the first sub-path of static links has weight greater than $m$. The weight of $f$ must be at least $\mu + m+1 +\mu +1= m+2+2\mu > m+1+2\mu$.
	\end{itemize}
\end{itemize}

Suppose that $f$ involves $1$ dynamic link; so, the structure of $f$ is $s^\ast d s^\ast$. If the dynamic link is an $\alpha$- or $\beta$-dynamic link then no matter which dynamic link this is, the weight of $f$ is greater than $3m-1+\mu > m+1+2\mu$. If the dynamic link is not an $\alpha$- or $\beta$-dynamic link then the weight of $f$ is at least $m+\mu+ 4 > m+1+2\mu$, as this dynamic link cannot be incident with any node of $W$. Finally, if $f$ does not involve a dynamic link then the weight of $f$ is at least $4m>m+1+2\mu$. Hence, every flow-path serving some $1$-demand has weight at least $m+1+2\mu$ and when it has weight exactly $m+1+2\mu$, it takes the form of a static path from $r$ to some port node of $P$, augmented with a dynamic link to some neighbour of an associate clause node, augmented with a static link to the associate clause node, and augmented with a dynamic link from the associate clause node to some element node. As the total workload cost of all flow-paths serving $1$-demands is at most $\kappa_1 = 3n(m+1+2\mu)$, every such flow-path has weight exactly $m+1+2\mu$ and is therefore of the form just described. 

There are $3n$ $1$-demands yet only $n$ dynamic links from a port node to some neighbour of an associate clause node. Hence, these $n$ neighbours of associate clause nodes must lie on $3n$ flow-paths. Consequently, these neighbours of associate clause nodes must actually be clause nodes. Let these clause nodes be $C^\prime = \{c_{j_1}, c_{j_2}, \ldots, c_{j_n}\}$. Thus, we have that there is a dynamic link from every associate clause node of each $c_{j_i}$ to some element node and we have a subset of clauses $\mathcal{C}^\prime = \{c_{j_1}, c_{j_2}, \ldots, c_{j_n}\}$ so that every element of $\mathcal{X}$ lies in exactly one clause of $\mathcal{C}^\prime$; that is, $(\mathcal{X},\mathcal{C})$ is a yes-instance of RXC3. The result follows as our instance $(H,\mu,D,\kappa)$ can clearly be constructed from $(\mathcal{X},\mathcal{C})$ is time polynomial in $n$.\qed \end{proof}

We emphasize the relevance of the hypercube as a prototypical model of interconnection networks (see, e.g., \cite{GLL09}) and the fact that we obtain hardness here for any choice of fixed dynamic link weight $\mu$ between $0$ and $1$. 

Taken together our results comprehensively establish the computational hardness of RRP in practically relevant settings. In particular, we establish that the problem remains intractable in several cases where the demand matrix is sparse, the hybrid network is highly structured (in fact node-symmetric) and the weights of links depend only on their medium. 

\subsection{The case of $\delta=1$}\label{sec:lunar}

This section is devoted to the restriction of RRP where the dynamic link limit $\delta$ is set to 1 - where any flow-path must use no more than a single dynamic link. We shall require substantially different techniques from those we have used up until now so to prove our main result in the remainder of the paper, which entails that the problem \textsc{1-switched RRP}($\sigma=1$) is NP-complete for various graph classes including hypercubes, grids and toroidal grids. 
The intention is that this section provides a convenient template for proving hardness of the problem for other practically interesting classes beyond the ones explicitly considered here.

\subsubsection{Additional definitions}

We begin with some basic definitions from graph theory. Let $G=(V,E)$ be a simple undirected graph. 
We define the \emph{open neighborhood} of a vertex $v$ to be $N(v):=\{u : (u,v) \in E\}$, and its closed neighborhood $N[v]:=N(v) \cup \{v\}$. Likewise for any set of vertices $S$, we define $N[S]:=\cup_{v \in S} N[v]$ and $N(S):=N[V] \setminus S$. 
A \emph{dominating set} in $G$ is a set of vertices $s$ such that each vertex in $G$ is either in $S$ or adjacent to a vertex in $S$. The \emph{domination number} of $G$, denoted $\gamma(G)$, is the least number of vertices in any dominating set of $G$. \textsc{Dominating Set} is the decision problem asking, for input $G$ and $k$, whether $\gamma(G) \le k$.

Where $S \subseteq V$ is a set of vertices in the graph, we denote $G[S]$ the subgraph of $G$ induced by $S$. That is, $G[S]$ has $S$ as its set of vertices and as edges exactly those edges of $G$ with both endpoints incident to a vertex in $S$. Also, $\mathrm{dist}_G(u,v)$ is the number of edges on the shortest path between $u$ and $v$ in $G$. 

\begin{definition}[Ball, Sphere]
    Where $G$ is a graph, we denote $B_G(v,r)$ (resp. $S_G(v,r)$) the \emph{ball} (resp. \emph{sphere}) with center $v$ and radius $r$ in $G$. Formally these are sets of vertices defined as:
    \begin{align*}
        B_G(v,r) &:= \{u : \mathrm{dist}_G(u,v) \le r\} \\
        S_G(v,r) &:= \{u : \mathrm{dist}_G(u,v) = r\} 
    \end{align*}
\end{definition}

\noindent We shall also make use of the following (likely non-standard) definitions.

\begin{definition}[Restriction of decision problems to a graph class]
    Let $\Pi$ be a decision problem which takes one or many inputs, exactly one of which is a simple undirected graph. We denote by $\Pi(\mathcal{G})$ the restriction of $\Pi$ to the graph class $\mathcal{G}$. That is, $\Pi(\mathcal{G})$ has as yes-instances (resp. no-instances) exactly those yes-instances (resp. no-instances) of $\Pi$ where the graph portion of the input belongs to $\mathcal{G}$. 
\end{definition}

We introduce a new graph problem, which is subtly different from the classic \textsc{Dominating Set}, and which has the property that it may remain hard even when the graph portion of the input is highly structured. This subtlety will become important presently - note that the property is precisely that which we wish to study for \textsc{1-Switched RRP}($\delta=1$).

\begin{framed}
    \noindent
    \textbf{\textsc{Partial Domination}}\\
    \emph{Input:} Simple undirected graph $G=(V,E)$; set $T \subseteq V$; integer $k$.\\
    \emph{Question:} Is $(G[T],k)$ a yes-instance of \textsc{Dominating Set}? Equivalently, is there some set $X\subseteq T$ with $|X| \le k$ such that $T \subseteq N_G[X]$?
\end{framed}


\subsubsection{Hardness of \textsc{Partial Domination} on (toroidal) grids and hypercubes}\label{subsubsec:pardom}

We now show that \textsc{Partial Domination} is NP-complete for several graph classes of interest to us, namely grids and hypercubes. This computational hardness (together with another graph class property we shall come to later) provides the foundation for our proof of Theorem \ref{thm:lunarRRP}.
We note that \textsc{Dominating Set} is trivially tractable for grids, though the same cannot be said of hypercubes; even $\gamma(Q_{10})$ is unknown \cite{oeis}. 
This subsection leverages several results from the literature, which more or less straightforwardly yield the desired results. 

\begin{theorem}[\cite{CCJ90} Theorem 5.1]
    \textsc{Dominating Set}(Induced subgraphs of grids) is NP-complete.
\end{theorem}

Note that the corollary below relies on some subtle properties of the proof applied in \cite{CCJ90} (namely, that the construction provided explicitly describes an embedding into some grid). 
It is in general NP-hard, given some graph, to determine whether it is an induced subgraph of a grid (and also to produce its vertices' coordinates in a grid - see \cite{GSZ23} and references therein).

\begin{corollary}
    \textsc{Partial Domination}(Grids) is NP-complete.
\end{corollary}

We now turn to the other graph class of interest to us - the hypercubes $\mathcal{Q}$. Although it is straightforward to show that every grid is the induced subgraph of some hypercube, we require for technical reasons that, more strongly, the hypercube in question is at most polynomially larger than the contained grid. Fortunately, we are able to rely here on the extensive literature surrounding the so-called \emph{snake-in-the-box} problem, which consists in finding large induced cycles in hypercubes. The following result belongs to that body of work.

\begin{theorem}[Abbott and Katchalski \cite{AK91,AK88}]\label{thm:ak-snake}
    For any $d\ge 2$, the hypercube $Q_d$ contains an induced path on $\frac{77}{256}2^d$ vertices. Given $d$, the coordinates of such a path may be produced in time polynomial in $2^d$.
\end{theorem}

In their work, Abbott and Katchalski state that there is an induced cycle on \textbf{strictly more than} $\frac{77}{256}2^d$ vertices, which entails an induced path on exactly $\frac{77}{256}2^d$ vertices.
The authors do not discuss the runtime of their construction in their work. However, it is clear from their proof that their description of the induced cycle can be realized as a recursive algorithm with a runtime as stated in the theorem. Their result has the following consequence, which shall be useful to us towards proving the computational hardness of \textsc{Partial Domination}($\mathcal{Q}$).

\begin{corollary}\label{cor:grid-in-hypercube}
    For any $d \ge 2$, the hypercube $Q_{2d}$ contains an induced grid on $(\frac{77}{256}2^d)^2$ vertices. Given $d$, the coordinates of such a grid may be produced in time polynomial in $2^d$.
\end{corollary}

\begin{proof}
    Given $d$, produce coordinates of a path $P=\{p_1, p_2, \ldots, p_\ell\}$ of length $\ell=\frac{77}{256}2^d$ in $Q_d$ by applying Theorem \ref{thm:ak-snake}. Note that each (node) $p_i$ is a bitstring of length $d$ exactly. Then the set $\{p_ip_j : i,j \in [\ell]\}$ is an induced grid of size $(\frac{77}{256}2^d)^2$ in $Q_{2d}$, and the result follows.
\end{proof}

It remains to combine the above results.

\begin{theorem}\label{thm:hypercube-partial-domination-np-c}
    \textsc{Partial Domination}($\mathcal{Q}$) is NP-complete.
\end{theorem}

\begin{proof}
    Let $(G,T,k)$ be an instance of \textsc{Partial Domination}(Grids). We assume without loss of generality that $G$ is an $n \times n$ grid (if necessary, by extending $G$ in some dimension and retaining the same set $T$).
    We shall denote each vertex in $G$ by $v_{i,j}$ with $i,j \in [n]$. 
    Let $d=\lceil \log_2(\frac{256n}{77})\rceil$. Applying Corollary \ref{cor:grid-in-hypercube}, we may efficiently produce a mapping from $V(G)$ (vertices of the $n \times n$ grid) to $V(Q_{2d})$. Denote this mapping $f$.
    Let $G'=Q_{2d}$, and $T'=\{f(t) : t \in T\}$. 
    Then $G'[T']=G[T]$ (so $\gamma(G'[T'])=\gamma(G[T])$ also) and $(G',T',k)$ is an instance of \textsc{Partial Domination}($\mathcal{Q}$). The construction of $G'$ and $T'$ is feasible in polynomial time, and the result follows. \qed
\end{proof}

A reader interested in showing that some other graph class $\mathcal{G}$ (e.g. subcubic graphs) is hard for \textsc{1-Switched RRP} may substitute the above for a proof that \textsc{Partial Domination}($\mathcal G$) is NP-complete. As briefly alluded to earlier, this is one of two properties we shall require of a graph class in the proof of Theorem \ref{thm:lunarRRP}; we turn to the second property presently. 

\subsubsection{Lunar graph classes}\label{subsubsec:lunar}

We have chosen to adopt a celestial metaphor to aid intuition, since in defining the class we make use of spheres, balls, large distances, and vast differences in size. The reader may find Figs. \ref{fig:lunar-grids} and \ref{fig:lunar-hypercube-weights} helpful in illustrating the definition. 

\begin{definition}[Moon, Planet, Sun]
    We say a graph $G$ is the \emph{moon} in some graph $H$ if:
    \begin{itemize}
        \item There is some set $M\subset V(H)$ with $H[M]=G$. We also call $M$ the \emph{moon} in $H$.
        \item There is some $o \in V(H)$ and integer $r$ such that the set $\surface:=S(o,r)$ is a sphere of cardinality at least $|M|$ and $\planet:=B(o,r)$ is the ball with the same center and radius. We call $\planet$ the \emph{planet} and $\surface$ the \emph{planet surface} (or, briefly, \emph{surface}) in $H$.
        \item $N[M] \cap N[\planet]=\emptyset$ (the planet and the moon are far apart).
        \item $|\surface| \ge |M|$ (the surface is bigger than the moon).
        \item There is some set $S \subset V(H)$ (the \emph{sun} in $H$) such that $|S| \ge |N[\planet]| + |N[M]|$ (the sun is bigger than the moon and planet together) and $N[S]$ is disjoint from $N[\planet]$ (the sun and planet are far apart). Moreover any node in $S$ is at distance at least $2r$ from any node in $M$ (the sun and moon are very far apart).
    \end{itemize}
\end{definition}

We continue with an astronomical theme for the naming of the graph class of interest itself:

\begin{definition}[lunar graph class]
    We say a graph class $\mathcal{L}$ is \emph{lunar} in a graph class $\mathcal{H}$ if, for each $G \in \mathcal{L}$, there exists some $H \in \mathcal{H}$ such that $G$ is the moon in $H$. 
    We further require that such a graph $H$ (and vertex sets $M,\surface,S$ within it) can be constructed in polynomial time from $G$ (which entails that $H$ is at most polynomially larger than $G$). We say a graph class $\mathcal{L}$ is \emph{lunar in itself} (or simply \emph{lunar}) if the above holds for $\mathcal{H}=\mathcal{L}$.
\end{definition}

For this definition to be useful, we still need to show that it holds for those graph classes which we are interested. The previously hypothesized reader interested in proving NP-completeness of \textsc{1-Switched RRP}($\delta=1$) restricted to, e.g., subcubic graphs, may find the following a useful blueprint to prove that subcubic graphs are lunar.

\begin{figure}[!ht]
    \centering
    \includegraphics[width=1\linewidth,page=10]{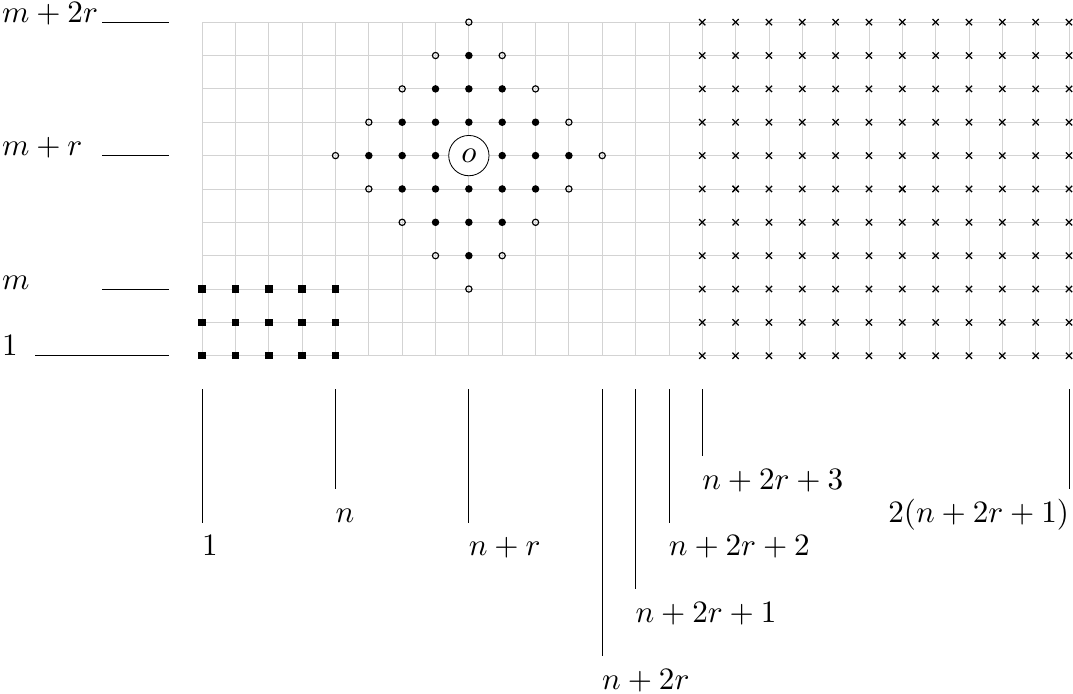}
    \caption{Illustration of our construction where $n=5, m=3, r=4$. The $28 \times 11$ grid contains the $5 \times 3$ grid as a moon. 
    Marked are: the moon $M$ (square vertices); $o=(n+r,m+r)=(7,9)$, together with the planet $\planet=B(o,r)$ ($o$ and disk and circle vertices) and its surface $\surface=S(o,r)$ (circle vertices); the sun $S=\{(i,j) : i \le 8\}$ (cross vertices).}
    \label{fig:lunar-grids}
\end{figure}

\begin{lemma}\label{lem:grids-are-lunar}
    The class of grid graphs is lunar in itself.
\end{lemma}
\begin{proof}
    The reader may find the illustrative example in Fig. \ref{fig:lunar-grids} helpful. Let $G$ be some $n \times m$ grid (w.l.o.g. $n \ge m$, so $|G|$ is polynomial in $n$). 

    We shall use the fact that a sphere of radius $r \ge 1$ in a grid (which does not spill over the grid's boundary) has size $4r$. 
    
    Let $r=\lceil \frac{nm}{4} \rceil$. 
    Let $x=2(n + 2r + 1)$ and let $y=2(m + 2r)$. We choose $H$ to be the $x \times y$ grid.
    Then we identify:
    \begin{itemize}
        \item The moon: $M=V(G)$ (vertices of the $n\times m$ grid are also vertices of $x \times y$ grid).
        \item The planet: $o=(n+r, m+r)$, $\surface=S(o,r)$, $\planet=B(o,r)$. 
        \item The sun: $S=\{(i,j) : i \ge n + 2r + 3\}$.
    \end{itemize}
    It is easy to verify that the neighborhoods of the moon, planet and sun are disjoint; that the sun and moon are at distance at least $2r$; that the size of the planet exceeds that of the moon; and that the size of the sun exceeds that of the moon and planet's neighborhoods combined. Since our construction may be carried out in polynomial time, the result follows.
\qed \end{proof}

The $x \times y$ grid is an induced subgraph of the $2x \times 2y$ toroidal grid, with the notable property that every shortest path in the former remains a shortest path in the latter. Applying this fact, the proof above can straightforwardly be adapted to obtain the following corollary:

\begin{corollary}\label{cor:torus-grids-are-lunar}
    The class of grid graphs is lunar in the class of toroidal grid graphs.
\end{corollary}

Note that the class of toroidal grid graphs is \emph{not} lunar in itself. We now turn to the protagonist of this paper, the class of hypercubes. 

\begin{lemma}\label{lem:hypercubes-are-lunar}
    The class of hypercubes $\mathcal{Q}$ is lunar in itself.
\end{lemma}

\begin{proof}    
    Note that we reuse some of the notation from our proof of Theorem \ref{thm:hypercubeRRPsigma3}.
    Let $G\in \mathcal{Q}$ be some hypercube of dimension $d$, i.e. $G=Q_d$ for some $d$. We define:
    \begin{itemize}
        \item Let $\ell=6d+9$ and $H=Q_{\ell}$.
        \item Let $M:=\{0^{5d+9}x : x \in Q_d\}$. 
        \item Let $o:=0^{4d+6}1^{2d+3}$ and $r=d$. Recall we denote the planet $\planet=B(o,r)$ and its surface  $\surface=S(o,r)$. 
        \item Let $S:=\{\overline{x} : x \in N[\planet] \cup N[M]\}$. Intuitively, the sun nodes are the reflection of the neighborhoods of planet nodes and moon nodes.
        
    \end{itemize}

    The illustration in Fig. \ref{fig:lunar-hypercube-weights} will be helpful in verifying the following:
    \begin{itemize}
        \item $H[M]=G$.
        \item The sphere $\surface=S(o,r)$ has cardinality at least $|Q_d|$ (easy to see by considering $\{o \oplus 0^{4d+9}\overline{x}{x} : x \in Q_d\}$, which clearly contains only vertices at distance exactly $d$ from $o$ and so is a subset of $\surface$).
        \item $N[M]$ and $N[\planet]$ are disjoint. 
        \item $N[M]$ and $N[\planet]$ contain only vertices with at most $3d+4$ ones. 
        \item $|S| \ge |N[\planet]| + |N[M]|$ (by applying both two points above).
        \item $N[S]$ is disjoint from $N[M]$ and $N[\planet]$.
        \item Moreover any node in $S$ is at distance at least $2r$ from any node in $M$. 
    \end{itemize}

    Note that $Q_\ell$ has size $2^{6d+9}=2^9\cdot (2^d)^6$ which is polynomial in the size of $Q_d$. The result follows.
    \qed
\end{proof}

\begin{figure}
    \centering
    \includegraphics[width=0.8\linewidth, page=9]{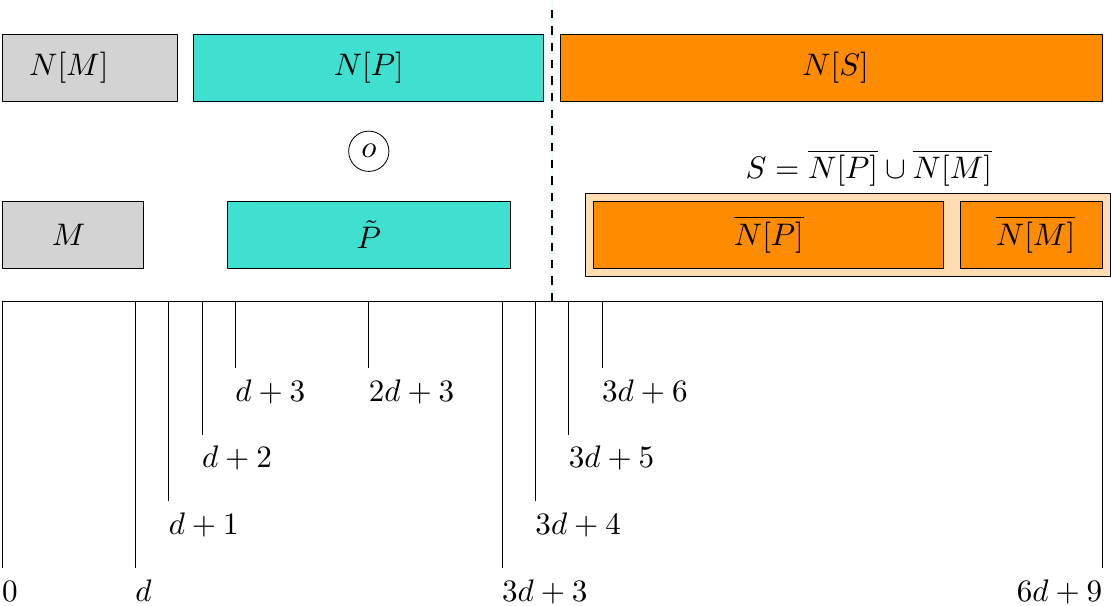}
    \caption{Illustration of our proof hypercubes are lunar, showing the weights (number of ones in any vertex label) for different sets of vertices. The midsection of the hypercube (separating majority-1 vertices from majority-0 vertices) is shown as a dashed line.}
    \label{fig:lunar-hypercube-weights}
\end{figure}

\subsubsection{The main result}\label{subsubsec:lunar-rrp}

We are now able to state and prove the main result of this subsection.

\begin{theorem}\label{thm:lunarRRP}
    For any fixed $\mu \in (0,1)$, and any graph classes $\mathcal{L}$ and $\mathcal{G}$ with $\mathcal{L}$ lunar in $\mathcal{G}$, the problem \textsc{Partial Domination}($\mathcal{L}$) is polynomially reducible to the problem \textsc{1-switched RRP($\delta = 1$)} restricted to instances $(H, \mu, D, \kappa)$ satisfying:
	\begin{itemize}
		\item $H\in \mathcal{G}$, and 
		\item the workload matrix $D$ is sparse and all values in it are polynomial in $|H|$.
	\end{itemize}
\end{theorem}

\begin{proof}
    We are given an instance $(G, T, k)$ of \textsc{Partial Domination}($\mathcal{L}$), with $G\in \mathcal{L}$, a set of vertices $T \subseteq V(G)$, and integer $k$. We shall produce an instance $(H, \mu, D, \kappa)$ of \textsc{1-Switched RRP}($\delta=1$).
    $\mu$ is some prescribed value between $0$ and $1$, as in the theorem statement.
    
    The graph $H$, along with vertex sets $M,\surface,\planet,S$, the vertex $o$, and the integer $r$, are obtained by applying the definition of a lunar class. We denote by $T'$ the set of vertices in $H$ to which $T$ is mapped, so that $H[M][T]=H[T']=G[T]$.

    In order to describe our demands $D$, we identify some useful sets of vertices in $H$:
    \begin{itemize}
        \item \emph{Target vertices} are exactly the set $T'=\{t_1',t_2',\ldots, t_{|T|}'\}$. 
        \item The set of moonlit nodes $\moonlit=\{\breve{m}_1,\breve{m}_2,\ldots,\breve{m}_k\}$ is some arbitrary subset of $\surface$ with cardinality $k$ exactly.
        \item The set of sunlit nodes $\sunlit:= N[M] \cup N[\planet] \setminus (T \cup \moonlit)$. 
    \end{itemize}

    The intention is that our construction will ensure that, in any configuration $N$ of interest to us, the moonlit nodes $\moonlit$ (resp. sunlit nodes $\sunlit$) will be connected by a dynamic link to moon nodes $M$ (resp. sun nodes $S$). 
    
    Consider a new graph: the complete bipartite graph with $\sunlit \cup S$ as vertices and $\sunlit \times S$ as edges. Let $E_S$ be an arbitrary maximum matching in this bipartite graph. Note that $|S| \ge |\sunlit|$ by construction, and so $|E_S|=|\sunlit|$ exactly and each sunlit vertex is incident to exactly one edge in $E_S$. Note that by construction $E_S \cap E(H) = \emptyset$, because none of the sunlit nodes are adjacent to any sun node (by applying a property of lunar graphs). That is, if $(u,v)$ is an edge in $E_S$ then $u$ and $v$ are not adjacent in $H$. 
    
    Let $\alpha=(r+\mu+1)|T|$. Our demand matrix $D$ is fully described by the following: 
    \begin{itemize}
        \item $D[u,v]=\alpha$ for each $(u,v)\in E_S$ (sunlight demands), 
        \item $D[o,t]=1$ for each $t\in T$ (moonlight demands), and
        \item all other entries of $D$ are zero.
    \end{itemize}
    
    \begin{figure}
        \centering
        \includegraphics[width=0.8\linewidth, page=8]{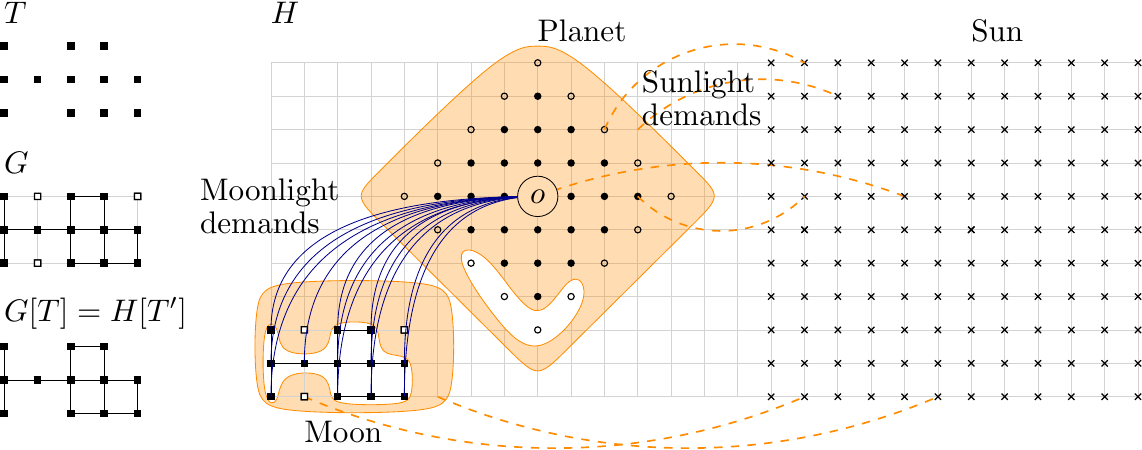}
        \caption{Illustration of the demands described in our reduction, for the instance $(G,T,4)$ of \textsc{Partial Domination}(Grid graphs). Nodes of $T$ (and $T'$ in $H$) are shown as black squares; other nodes of $G$ (and $M$ in $H$) are shown as boxes. Other nodes are shown as in Fig. \ref{fig:lunar-grids} earlier.
        Moonlight demands are shown as solid blue curves. The sunlit nodes $\sunlit$ are those in orange shaded regions (note exactly $4$ nodes of $\surface$ are not sunlit - these are the moonlit nodes $\moonlit$). Sunlit demands are drawn as orange arcs (for clarity, only a few are shown).}
        \label{fig:lunar-construction-demands}
    \end{figure}

    It remains for us to define $\kappa$. As before, we define this as $\kappa_\alpha + \kappa_1$, with $\kappa_\alpha = \mu \alpha |E_S|$ and $\kappa_1 = |T|(r + \mu) + |T| - k$. We note that $\alpha > \kappa_1$.
    This completes our construction of the instance $(H, \mu, D, \kappa)$.

    \begin{claim}
        If $(G,T,k)$ is a yes-instance of \textsc{Partial Domination} then $(H, \mu, D, \kappa)$ is a yes-instance of \textsc{1-Switched RRP}($\delta=1$).
    \end{claim}
    \begin{proof}
        An illustration of an optimal configuration is shown in Fig. \ref{fig:optimal-lunar-configuration}.
        Let $X$ be a dominating set of $G[T]$ of cardinality $k$. Denote also by $X'=\{x_1',x_2', \ldots, x_k'\}$ the corresponding set of vertices in $H$. 
        
        Let $N := E_S \cup \{(x_i', \breve{m}_i) : 1 \le i \le k\}$. Clearly, under $N$ each sunlight demand is served at cost exactly $\mu \alpha$ and so all sunlight demands cumulatively are served at cost $|E_S|\mu\alpha=\kappa_\alpha$. 

        Further, the $k$ demands from each node $x_i'$ to $o$ are each served at cost exactly $\mu+r$ (by the path $x_i'\dynamic \breve{m}_i \rightsquigarrow_{r} o$), and the moonlit demand for each of the $|T|-k$ other nodes $t\in T$ is served at cost exactly $1 + \mu + r$ (by the path $t \rightsquigarrow_1 x_i' \dynamic \breve{m}_i \rightsquigarrow_{r} o$). The claim follows.\qed
    \end{proof}

    \begin{figure}[!ht]
        \centering
        \includegraphics[width=0.6\linewidth, page=11]{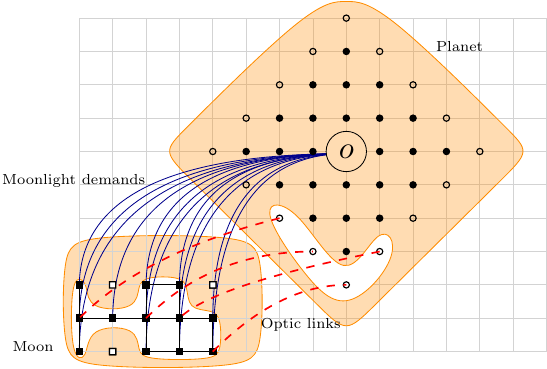}
        \caption{Detail of an optimal configuration $N$for the instance shown in \ref{fig:lunar-construction-demands}. Only moonlight demands (solid blue arcs) and dynamic links serving them (dashed red lines) are shown. Note that all moonlit demands are served at cost exactly $\kappa_1$.}
        \label{fig:optimal-lunar-configuration}
    \end{figure}

    Conversely, suppose that $(H, \mu,D,\kappa)$ is a yes-instance of \textsc{1-switched RRP\linebreak ($\delta=1$)} and that $N$ is a configuration and $F$ a set of flow-paths witnessing that the total workload cost is at most $\kappa$. 
	
	\begin{claim}\label{clm:all-sunlit}
        Each sunlit node has a line of sight (optic link) to a sun node.	
        Formally, every edge $(u,v)$ of $E_S$ is necessarily such that $(u,v)\in N$ and the total workload cost of flow-paths serving the $\alpha$-demands is exactly $\kappa_\alpha$.        
	\end{claim}

    \begin{proof}
    	Suppose that at least one of the $\alpha$-demands is served via a flow-path at a workload cost of more than $\mu\alpha$; so, it must be at a workload cost of more than $(1+\mu)\alpha$ as we cannot traverse a dynamic link followed immediately by another dynamic link (recall, $\Delta_S=1$) and by our construction, if $D[u,v]=\alpha$ then there is no static link $(u,v)$. Hence, the total workload cost of flow-paths serving the $\alpha$-demands is at least $(|E_\alpha| - 1)\mu\alpha + (1+\mu)\alpha = \kappa_\alpha +\alpha$. We have that $\kappa_1 < (r+\mu+1)|T| = \alpha$. Hence, the total workload cost of all flow-paths of $F$ is strictly greater than $\kappa$ which yields a contradiction and the claim follows.
    \qed \end{proof}

    \begin{claim}\label{clm:moonlit-demands}
        Each moonlight demand is served at cost at least $\mu + r$. Furthermore, at most $k$ moonlight demands are served at cost $\mu + r$ exactly. 
    \end{claim}
    \begin{proof}
        Consider some moonlight demand; necessarily it has form $D[t,o]=1$ for some $t\in T$.
        
        We first show each moonlight demand is served at cost at least $\mu+r$. Suppose for contradiction that $D[t,o]$ is served at cost strictly less than $\mu + r$. Since $t$ and $o$ are at distance at least $r+3$, the flow-path from $t$ to $o$ must make use of some dynamic link $u \dynamic v$ at cost $\mu$. The static portion of the path therefore must have cost at most $r-1$, which entails that $v \in B(o,r-1) \subsetneq \sunlit$ (that is, $v$ is a sunlit node). Applying Claim \ref{clm:all-sunlit} $v$ is connected by dynamic link to some sun node, so $u \in S$ necessarily. Then the path from $t$ to $u$ has length at least $2r$, yielding the desired contradiction. 

        Now observe that $D[t,o]$ can be served at cost exactly $\mu + r$ if and only if the dynamic link $t \dynamic \breve{m}$ exists, for some $\breve{m} \in \moonlit$. Since $|\moonlit|=k$ exactly by construction, we obtain that this may be the case for at most $k$ nodes in $T$. The claim follows. 
        \qed
    \end{proof}

    Since $\kappa_1 = (r+\mu)(k) + (r + \mu + 1)(|T|-k)$, we immediately obtain that \emph{exactly} $k$ moonlight demands are served at cost $r+\mu$ and all remaining moonlight demands are served at cost $r+\mu+1$ exactly.
        
    \begin{claim}\label{clm:all-moonlit}
        The set $X=\{ u: (u,v)\in N $ and $v\in \breve{M}$ is a moonlit node$\}$ is a dominating set of $H[M]$.
    \end{claim}
    \begin{proof}
        Suppose for contradiction that there is some vertex $t\in T'$ which is neither in $X$ nor adjacent to any vertex in $X$. We show that the cost of serving the demand $D[y,o]=1$ is strictly greater than $r + \mu + 2$. Denoting arbitrary nodes $m'\in N[M],\planetv' \in N[\planet], \surfacev \in \surface$, this flow is routed either:
        \begin{itemize}
            \item Via static links only, along a path of length at least $r+3>r+\mu+1$:
            $y\rightsquigarrow_{\ge1} m' \rightsquigarrow_{\ge1} \planetv' \rightsquigarrow_1 \surfacev \rightsquigarrow_r o$, or
            \item Via static links and one dynamic link $u\dynamic v$ with:
            \begin{itemize}
                \item $v$ in $\surface$, at cost at least $r+\mu+2>r+\mu+1$: 
                $y \rightsquigarrow_{\ge 2} u \dynamic v \rightsquigarrow_r o$
                \item $u$ and $v$ both outside $N[\planet]$, at cost at least $r+\mu+2>r+\mu+1$: 
                $y \rightsquigarrow_{\ge 0} u \dynamic v \rightsquigarrow_{\ge 1} \planetv' \rightsquigarrow_1 \surfacev \rightsquigarrow_r o$
                \item either $u$ or $v$ in $N[\planet] \setminus \surface$ (and the other in $S$), at cost at least $2r + \mu$ (recall all vertices in $S$ are distance at least $2r$ from any vertex in $M$).
            \end{itemize}
        \end{itemize}
        This contradicts our earlier claim (that all moonlight demands are served at cost at most $r+\mu+1$) and the result follows.\qed
    \end{proof}
    We note that the construction described takes polynomial time, and the main result follows.
\qed    
\end{proof}

Applying our earlier results on lunar graph classes (Lemmas \ref{lem:hypercubes-are-lunar} and \ref{lem:grids-are-lunar}, Corollary \ref{cor:torus-grids-are-lunar}) together with Theorem \ref{thm:hypercube-partial-domination-np-c} we obtain that RRP remains hard even when the number of dynamic links admitted on any path is limited to $1$.

\begin{corollary}
    \textsc{1-switched RRP($\delta = 1$)($\mathcal{G}$)} is NP-complete if $\mathcal{G}$ is: the class of hypercubes;  the class of grid graphs; or the class of toroidal grid graphs.
\end{corollary}

We emphasize again the value of hypercubes as a prototypical model of interconnection networks, and additionally note that grids and toroidal grids exhibit additional properties which may have been expected to yield a tractable setting, such as planarity and bounded degree. 
The restriction to $\delta=1$ in our setting also implies a restriction to $\sigma=2$ (as 3 alternations along a path would entail a minimum 2 dynamic links along the same) but the converse does not hold; we expect that the case of \textsc{1-Switched RRP}($\sigma=2, \delta=2$) can be shown to be intractable through similar proof techniques, but leave this for future work.

\section{Discussion and Future Work}

Taken together, our results comprehensively establish the computational hardness of RRP in practically relevant settings. We establish that the problem remains intractable in several cases where the demand matrix is sparse, the hybrid network is highly structured (in fact node-symmetric) and the weights of links depend only on their medium. Furthermore, in all of our hardness results, the instrument used to ``express'' NP-completeness is the demand matrix $D$. In the real world, the computational workload for the network is generally expected to vary significantly with time, unlike the network's hardware, which (in addition to its structural properties already discussed) does not rapidly change. Our results are in this sense closely relevant to the hardness of the real world reconfigurable routing problem.


We take this opportunity to identify some specific questions we have left open, as well as several more general avenues for future work in this area. 
First, it would be interesting to study the restriction of the problem to cases where $\mathrm{\Delta}_S$ is greater than $1$ and $\mu$ is a fixed constant. Results in this setting would ``bridge the gap'' between Theorems \ref{thm:2switchedRRP} and \ref{thm:3switchedRRP}, and Theorems \ref{thm:hypercubeRRPsigma3} and \ref{thm:lunarRRP}. Analogously, there is a gap for $1$-\textsc{Switched RRP} on hypercubes between $\sigma=0$ (which is solvable in polynomial time) and $\sigma=3$ (which is an intractable case). The complexity of the problem with $\sigma=1$ and $\sigma=2$ remains open for hypercubes (note that results for arbitrary networks do exist when $\sigma=2$, as shown in Table \ref{table:results}). Our Theorem \ref{thm:lunarRRP} in some sense sits between these two open cases, since the restriction to $\sigma=1$ entails a restriction to $\delta=1$, which itself entails a restriction to $\sigma=2$. 

Secondly, the present work considers only exact computation. In \cite{FFSV20} the authors establish inapproximability within $\Omega(\log n)$ for RRP in a more permissive setting (making use of variable link weights). However, the empty solution (there are no dynamic links and all demands are routed through the static network only) is a $\frac{\log n}{\mu}$-approximation for $\mathrm{\Delta}_S$-\textsc{Switched RRP} on hypercubes. (This follows straightforwardly from hypercubes having logarithmic diameter.) It would be interesting to see what (in)approximability results can be derived in our model with fixed link weights, with and without restrictions to realistic topologies. 

Lastly, parameterized algorithms may provide more fine-grained insights into the computational complexity of reconfigurable routing. Our Theorems \ref{thm:2switchedRRP} and \ref{thm:3switchedRRP} establish that structural parameters of the static network, such as treewidth, are insufficient to yield fixed-parameter tractable (fpt) algorithms (unless P=NP).
However, it would be interesting to see whether it is possible to obtain an fpt algorithm by additionally parameterizing by the sum of the demand matrix $D$; some structural parameters for the digraph representation of the demands, $D^\prime$;  the dynamic link weight ${\mu}$; or a combination of these.

\bibliography{RRP_algowin}{}
\bibliographystyle{plain}

\end{document}